\documentclass{article} 
\usepackage{iclr2027_conference,times}

\usepackage{amsmath,amsfonts,bm}

\def\eqref#1{equation~\ref{#1}}

\def\1{\bm{1}}

\DeclareMathAlphabet{\mathsfit}{\encodingdefault}{\sfdefault}{m}{sl}
\SetMathAlphabet{\mathsfit}{bold}{\encodingdefault}{\sfdefault}{bx}{n}

\usepackage{amsmath,amssymb,amsthm,mathtools}
\usepackage{algorithm}
\usepackage{algpseudocode}
\usepackage{array}
\usepackage{booktabs}
\usepackage{tabularx}
\usepackage{wrapfig}
\usepackage{enumitem}
\usepackage{graphicx}
\usepackage{float}
\usepackage{url}
\usepackage{tikz}
\usetikzlibrary{arrows.meta,positioning,fit,backgrounds}
\usepackage{hyperref}
\definecolor{cpunrow}{RGB}{231,241,248}
\definecolor{cpunfallback}{RGB}{211,229,242}

\newtheorem{theorem}{Theorem}
\newtheorem{lemma}{Lemma}
\newtheorem{proposition}{Proposition}
\newtheorem{corollary}{Corollary}
\newtheorem{definition}{Definition}
\newtheorem{assumption}{Assumption}

\newtheorem{remark}{Remark}

\newcommand{\Cn}{\mathrm{Cn}}
\newcommand{\Dd}{D_{\mathrm{d}}}
\newcommand{\Ddc}{D_{\mathrm{d}}^{\mathrm{cert}}}
\newcommand{\drec}{d_{\mathrm{rec}}}
\newcommand{\fc}{f_{\mathrm{c}}}
\newcommand{\Sem}{\mathrm{sem}}
\newcommand{\Ess}{\mathrm{Ess}}
\newcommand{\Atom}{\mathrm{A}}
\newcommand{\SO}{\mathbb{S}_{O}}
\newcommand{\PS}{\mathsf{PS}}
\newcommand{\TPS}{T_{\mathsf{PS}}}
\newcommand{\Dexp}{D_{\mathrm{exp}}}
\newcommand{\Scert}{S_{\mathrm{cert}}}
\newcommand{\Pcert}{\PS_{\mathrm{cert}}}
\newcommand{\epsfr}{\varepsilon_{fr}}
\newcommand{\bans}{\beta_{\mathrm{ans}}}
\newcommand{\slipc}{\sigma_{c}}
\newcommand{\err}{\mathrm{err}}
\newcommand{\PAu}{P_A^{u}}
\newcommand{\tr}{\mathrm{tr}}

\newcommand{\Core}{\mathcal{C}}
\newcommand{\Bridge}{\mathcal{B}}
\newcommand{\Carr}{\mathcal{N}}

\newcommand{\Query}{\textsc{Query}}
\newcommand{\Admit}{\textsc{Admit}}
\newcommand{\Evict}{\textsc{Evict}}
\newcommand{\Erase}{\textsc{Erase}}
\newcommand{\Tick}{\textsc{Tick}}
\newcommand{\key}{\mathrm{key}}

\title{CPUNeSy: Controlling Model Writes for Reliable Neuro-Symbolic Reasoning}

\author{
\textbf{Zeyan Li\textsuperscript{1} \quad
Siyuan Qiu\textsuperscript{1} \quad
Shuai Zhao\textsuperscript{1} \quad
Jianfeng Xu\textsuperscript{1}
}\\
\textsuperscript{1}Shanghai Jiao Tong University
}

\iclrfinalcopy

\begin{document}
\flushbottom

\maketitle

\begin{abstract}

LLMs excel at recalling statistical patterns but degrade sharply when answers must be derived, especially on multi-hop chains. Delegating derivation to deterministic symbolic executors shifts reliability to whether model-generated premises are source-supported. We introduce CPUNeSy, a serving architecture that controls model writes to symbolic state via a task-defined predicate interface and certificate gate, abstaining when grounding passes disagree. Component analysis isolates deterministic execution, restricted grounding, agreement, and source rechecking. Experiments show deterministic execution drives most accuracy recovery on derivation-heavy tasks; controlled writes mainly improve selective reliability by withholding unsupported or inconsistent answers, at a coverage cost. On multi-hop tests in law and formal math, deterministic execution recovers most of the gap over chain-of-thought and retrieval baselines, with full-pool gains up to 35.0 points. Certification is selective-serving control, not accuracy mechanism: with grounding traces fixed on ContractNLI, source rechecking removes a quarter of DeepSeek's wrong answers surviving two-vote agreement, at measurable coverage cost. When abstention is costly, routing withheld cases to an uncertified same-model fallback raises full-pool accuracy on MedCalc-Bench Verified by 13.9 and 4.9 points for Seed and DeepSeek; these gains are not from the certified channel. On LeanDojo Benchmark 4, kernel-restricted pools match BM25 recall@15 (89.3\%). Gains depend on the grounder's error regime: bias-dominated grounders benefit less, consistent with our voting bound. Certificates guarantee derivational validity relative to admitted premises; semantic faithfulness to natural-language sources remains conditional on the source checker, and prospective validation is future work.
\end{abstract}

\section{Introduction}
\label{sec:intro}
Large language models recall facts remarkably well \citep{kaplan2020scaling,srivastava2022bigbench}, yet on tasks whose correctness depends on \emph{composing} several rules or facts from a fixed source of truth, token generation often imitates a derivation rather than executing one \citep{dziri2024faith,kordjamshidi2026reasoners}. On real mathematical dependency graphs and multi-hop legal reasoning, chain-of-thought and retrieval baselines plateau well below what the underlying formal structure allows, and increasing the reasoning budget or the number of sampled traces yields diminishing returns \citep{weietal2022cot,lewis2020rag,wang2023selfconsistency}. The central difficulty is ensuring that an answer follows from the stated premises: a chain-of-thought trace can appear coherent even when it skips a necessary inference or introduces an unsupported intermediate conclusion. Legal, medical and formal-mathematics deployments make this concrete, since a downstream decision cannot rest on an answer whose intermediate steps are neither auditable nor tied to source.

A natural response is to delegate derivation to a symbolic executor once the model has proposed the relevant facts and rules \citep{logiclm,gao2023pal}. Execution is deterministic and cheap to check, and it turns a multi-hop question into a query against a well-defined logical state. But execution alone is not a correctness certificate: a valid derivation from a premise that misrepresents the source produces a wrong answer just as convincingly as a hallucinated chain (Figure~\ref{fig:overview}). A model can arrive at the same wrong answer either by fabricating an intermediate conclusion in a chain of thought or by writing it into the solver's premise set and letting the executor rediscover it. To rule this out, the system needs to decide which model-proposed premises may enter its state at all, and to keep evidence for every accepted premise so that a served answer can be traced end-to-end.

\begin{figure}[t]
\centering
\includegraphics[width=\columnwidth]{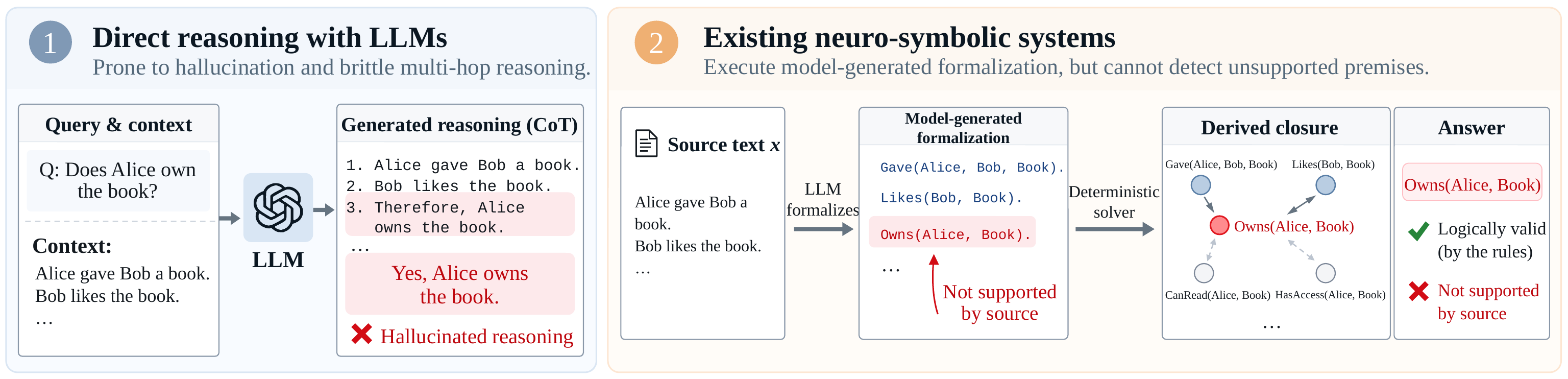}
\caption{Two routes to the same wrong answer: hallucinated reasoning
(\emph{left}) and valid execution of unsupported premises
(\emph{right}). These failure modes motivate controlling what language models may write into symbolic state.}
\label{fig:overview}
\end{figure}

Neuro-symbolic architectures combine learned and symbolic components \citep{garcez2022,nesy2024}, and recent LLM-plus-solver pipelines \citep{olausson2023linc,ye2023satlm,jiang2024leanreasoner,sun2024determlr,xu2024symbcot,li2024lina,li2025hblr,mental} inherit the correctness of their solver but expose it to whatever the parser writes: admission is implicit, and an ill-supported premise is checked only by whether the solver can still run. CPUNeSy (Certified Predicate-Unit Neuro-Symbolic Networks) treats admission as a first-class step (Figure~\ref{fig:method}). Under the architecture's contract, the model proposes predicate units, each a fact or a Horn rule; a certificate gate checks their source support and rejects proposals whose target predicate lies outside the allowed source vocabulary; a deterministic executor derives the answer; and every certified answer carries a bounded-depth derivation replayable against the current state. Restricting the source vocabulary prevents the model from inserting an intermediate conclusion, or the target answer, as a source fact, so the executor is the only route by which a rule-head predicate can enter state. 

Studying such a system requires disentangling several mechanisms that change at once when execution is added on top of a language model: derivation shifts from sampling to solving, the model is queried in a more structured way, and repeated calls are aggregated by voting. Our experiments separate these. We compare a language model alone, the same model paired with a deterministic solver, and the restricted two-vote policy, separating the contribution of execution from the additional serving policy. We then hold both grounding traces fixed on the same benchmark and compare plain two-vote agreement against agreement plus source rechecking, so the difference quantifies what source checking adds over voting itself. Finally, an offline admission replay tests the vocabulary barrier with injected target predicates; the historical restricted runs enforced that vocabulary in the prompt. Five suites span legal reasoning, formal mathematics, and medical calculation; controlled multi-hop tests measure execution under adversarial mid-link-break negatives, public ContractNLI and MedCalc comparisons expose how admission trades errors against abstentions, and a matched LeanDojo comparison tests kernel-based premise pools against BM25 and a learned retriever. We report answered accuracy, coverage, and full-pool accuracy that counts every abstention as an error. Two findings organize the analysis. First, deterministic execution accounts for most of the accuracy recovery on derivation-heavy tasks, so a one-vote solver already closes most of the gap. Second, controlled writes primarily improve selective reliability, withholding unsupported or inconsistent answers at the cost of coverage; their value is the error--abstention tradeoff they impose on top of execution, not a further accuracy gain. The resulting certificates are for derivational validity relative to admitted premises, and semantic faithfulness to natural-language sources remains conditional on the source checker. Our contributions are as follows:
\begin{itemize}[leftmargin=1.2em,itemsep=0pt,topsep=1pt]
  \item A checked interface that separates model proposals from symbolic state updates, with formal guarantees conditional on source-checker and verifier soundness.
  \item A component analysis separating deterministic execution from restricted grounding, agreement, and source rechecking, including a fixed-trace gate ablation on public ContractNLI data.
  \item An empirical characterization of the resulting error--coverage--cost tradeoffs, with diagnostics for when controlled writes and optional kernelization are useful.
\end{itemize}
\section{Related Work}
\label{sec:related}

\textbf{LLM-to-symbolic solver pipelines.}
The natural way to combine a language model with a symbolic executor is to let the model translate and the solver execute. Logic-LM \citep{logiclm}, LINC \citep{olausson2023linc}, and SatLM \citep{ye2023satlm} established this pattern for first-order logic, SAT, and constraint solving, and LeanReasoner \citep{jiang2024leanreasoner} and CLOVER \citep{ryu2025clover} carry it into proof assistants and verification-backed domains. Once translation worked often enough to be useful, the bottleneck moved to its failures. One-shot formalization is brittle, so DetermLR \citep{sun2024determlr}, LINA \citep{li2024lina}, and HBLR \citep{li2025hblr} wrap it in validate-and-repair loops, and SymbCoT \citep{xu2024symbcot}, Aristotle \citep{xu2025aristotle}, and CaRing \citep{yang2025caring} keep the reasoning trajectory itself in symbolic form so that errors can be caught mid-chain. Across these refinements, one thing did not change. Once a formal problem is written down, execution is trusted, and whether the written premises are supported by the source is never asked. \textsc{CPUNeSy} asks exactly that question, checking source support before a proposed predicate may become a premise for execution.

\textbf{Neuro-symbolic engines, provenance, and structured generation.}
A second body of work brings symbolic structure into learning rather than serving. DeepProbLog \citep{deepproblog}, Logic Tensor Networks \citep{ltn}, NeurASP \citep{neurasp}, Scallop \citep{li2023scallop}, LNN \citep{lnn}, and SATNet \citep{satnet} each make a different fragment of logic differentiable or probabilistic, so that gradients can flow through constraints during training. Alongside them, constrained decoding made formal output syntactically reliable, first for SQL with PICARD \citep{scholak2021picard} and then as general infrastructure with XGrammar \citep{dong2024xgrammar}, and Datalog provenance \citep{bourgaux2022provenance} tracks which base facts support each derived fact. 


\textbf{Neural theorem proving.}
Theorem proving followed a similar arc. GPT-f \citep{gptf} showed that a language model can predict useful proof steps, retrieval made premise selection practical at library scale through Magnushammer \citep{magnushammer} and LeanDojo's ReProver \citep{leandojo}, and COPRA \citep{thakur2023copra} turned proving into an agentic search over tactic applications. Recent work also verifies intermediate natural-language claims \citep{liu2025safe}. Our premise pools borrow the premise-retrieval setting but change the object being retrieved, since the pool is built from the source kernel and its dependency closure rather than a raw library. Section~\ref{subsec:rq1} compares them with ReProver and BM25 directly.

\textbf{Selective prediction, verification, and benchmarking.}
Knowing when not to answer has its own history. Selective prediction lets a model abstain below a confidence threshold \citep{geifman2017selective}, and post-hoc verification scores or repairs answers after generation \citep{cobbe2021verifier,liu2025safe}. A growing set of benchmarks, LogicBench \citep{parmar2024logicbench}, ProverGen/ProverQA \citep{qi2025provergen}, RuleArena \citep{zhou2025rulearena}, SLR \citep{helff2026slr}, LogicGraph \citep{wu2026logicgraph}, MetaPCR \citep{galitsky2026metapcr}, and Reasoners \citep{kordjamshidi2026reasoners}, exists because the community kept discovering that fluent systems reason incorrectly. \textsc{CPUNeSy} shares the abstention instinct but moves the check earlier, before execution rather than after generation, and the fixed-trace comparison in \S\ref{subsec:rq2} separates what source checking adds from what voting already provides. Appendix~\ref{app:related} gives the full comparison.
\section{The CPUNeSy Method}
\subsection{Architecture}
\label{sec:arch}
A \emph{predicate unit} is either a ground fact or a Horn rule. CPUNeSy uses these units to represent both source premises and derived conclusions. A model proposes units and inference steps, and a checker decides which proposals may enter the state. This addresses two common failure modes. A model that produces an answer and rationale together may skip a required inference. A model that writes a formal problem for a solver may introduce facts or rules unsupported by the source. Figure~\ref{fig:overview} illustrates the two failure modes, while Figure~\ref{fig:method} shows how CPUNeSy separates model proposals from checked state updates.

A CPUNeSy instance has three components. The language model $\Carr$ maps source text to facts and rules, a step called \emph{grounding}, and verbalizes the results. The checking interface $\Bridge$ extracts and verifies inference steps. The symbolic core stores a \emph{kernel}, a set of source units with no redundant members, and a cache of verified conclusions. We refer to the model and checking interface as the carrier and bridge in the formal definitions. The architecture specifies two phases (Figure~\ref{fig:method}). During \textsc{Deploy}, candidate facts and rules are checked for type and source support. A deletion scan in a fixed order removes each unit that can be derived from the remaining units. During \textsc{Serve}, a query is grounded against that kernel, candidate derivations are checked edge by edge, and only a certified result is verbalized. For example, if the checked source contains $p$ and $p\!\rightarrow\!q$, the system may derive and serve $q$ with both units in its certificate. The source kernel remains $\{p,p\!\rightarrow\!q\}$, from which $q$ can be reconstructed.

\begin{figure}[ht]
\centering
\includegraphics[width=\columnwidth]{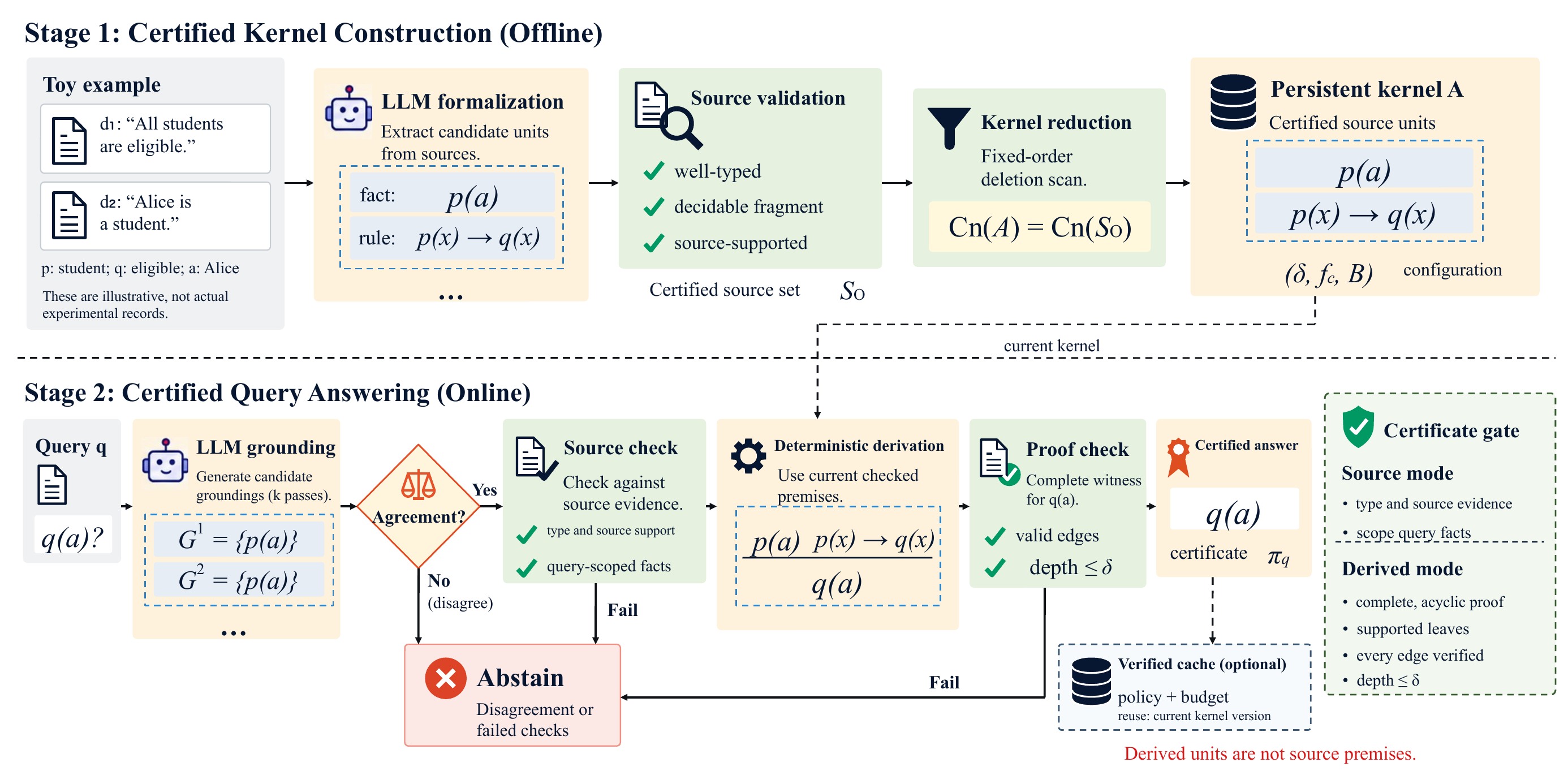}
\caption{The CPUNeSy architecture contract. \emph{Stage 1 (offline):} \textsc{Deploy}
formalizes and validates source units, then reduces the certified source set $S_O$ to an irredundant kernel $A$ with $\Cn(A)=\Cn(S_O)$.
\emph{Stage 2 (online):} \textsc{Serve} checks agreeing query groundings, derives an answer, and verifies its proof before release.}
\label{fig:method}
\end{figure}

The certificate gate uses separate checks for source premises and derived conclusions. In \textsc{Source} mode the gate checks a proposed fact or rule against its source record and adds supported units to a certified source set; a canonical deletion scan then recomputes the kernel and omits units that the remainder entails, which matters because a later source unit can make an earlier one redundant. In \textsc{Derived} mode the gate checks that every edge of a proposed derivation is valid over the current kernel and that its certified depth is within budget; a verified conclusion may then be served or cached while retaining its dependencies on the source premises. Repeated grounding filters proposals on which the model's calls disagree, and Algorithm~\ref{alg:gate} states both modes together. Table~\ref{tab:checker-contract} then specifies the checks used in each domain: the mathematical experiments execute recorded proof-library dependency graphs, while the natural-language tasks rely on schema- and vocabulary-grounded checks whose extraction step is model-mediated. The offline admission replay enforces the leaf-predicate restriction and preserves all recorded answers from that interface (Appendix~\ref{app:gate-audit}).

\begin{table}[t]
\centering\footnotesize
\caption{Source checks used in experimental adapter and extraction errors they can leave unresolved.}
\label{tab:checker-contract}
\begin{tabularx}{\linewidth}{@{}lXX@{}}
\toprule
Adapter & Source checks & Possible extraction errors \\
\midrule
L1-HARD & Rule-derived leaf vocabulary; parsed facts; agreement of solver answers & LLM maps narrative to leaf facts; historical vocabulary restriction was in prompt \\
ContractNLI & Schema-grounded facts; LLM source recheck of contributing facts; two verdicts & Rechecker can share the grounder's semantic bias \\
MedCalc & Calculator schema, types and units; equality of normalized attribute maps; executable calculation & Schema-valid values can still be incorrectly extracted \\
M-HARD / LD-MH & Candidate premise IDs; dependency closure in a frozen proof-library DAG & Mapping to IDs and completeness of the recorded graph \\
\bottomrule
\end{tabularx}
\end{table}
The experiments instantiate different parts of this contract. Historical L1-HARD runs restricted the vocabulary in the prompt; an explicit admission barrier is evaluated separately by offline replay. Complete per-query certificate archival was not instrumented in every adapter. Table~\ref{tab:disposition} and Appendix~\ref{app:gate-audit} distinguish these implementation scopes.

\begin{algorithm}[!t]
\caption{Certificate gate with separate source and derived modes}
\label{alg:gate}
\footnotesize
\begin{algorithmic}[1]
\Require mode $m\in\{\textsc{Source},\textsc{Derived}\}$; proposed unit $u$ with evidence $\xi$; certified source set $S_t$; core $\Core_t=(A_t,\mathcal{R}_t,\mathrm{Cache}_t)$; depth budget $\delta$; source checker $V_{\mathrm{src}}$; edge extractor $E$; fragment verifier $V$
\Ensure \texttt{SOURCE-ADMIT}, \texttt{REDUNDANT}, \texttt{CERTIFY}, or \texttt{REJECT}, with a checkable record
\State $\Gamma_t\gets A_t\cup\mathcal{R}_t$
\If{$m=\textsc{Source}$}
  \If{$u$ is ill-typed \textbf{or} $V_{\mathrm{src}}(u,\xi)=\bot$}
    \State \Return \texttt{REJECT} (invalid or unsupported source unit)
  \EndIf
  \State $\pi_{\mathrm{src}}(u)\gets\langle u,\xi,V_{\mathrm{src}}\text{-log}\rangle$
  \State $S_{t+1}\gets S_t\cup\{u\}$;\quad $(K_{t+1},J_{t+1})\gets\textsc{KernelScan}(S_{t+1})$
  \State partition $K_{t+1}$ into facts $A_{t+1}$ and rules $\mathcal{R}_{t+1}$
  \State log $\pi_{\mathrm{src}}(u)$; recheck or rebuild cache certificates over $K_{t+1}$ within $\delta$; evict failures
  \If{$u\in J_{t+1}$}
    \State \Return \texttt{REDUNDANT} with $\pi_{\mathrm{src}}(u)$ \Comment{certified, omitted from core}
  \EndIf
  \State \Return \texttt{SOURCE-ADMIT} with $\pi_{\mathrm{src}}(u)$
\EndIf
\State $\mathcal{E}(u)\gets E(u,\xi)$ \Comment{$\xi$ is the carrier trace}
\If{$\mathcal{E}(u)$ is not a complete acyclic derivation of $u$ over $\Gamma_t$}
  \State \Return \texttt{REJECT} (unlinked conclusion or unsupported leaf)
\EndIf
\ForAll{$e\in\mathcal{E}(u)$ in topological order}
  \If{$V(e,\Gamma_t)=\bot$}
    \State \Return \texttt{REJECT} (unverifiable derivation edge $e$)
  \EndIf
\EndFor
\State $d\gets\operatorname{depth}(\mathcal{E}(u))$ \Comment{depth of the verified witness}
\If{$d>\delta$}
  \State \Return \texttt{REJECT} (depth budget exceeded)
\EndIf
\State $\kappa\gets\key(u)$;\quad $\pi_{\mathrm{der}}(u)\gets\langle\mathcal{E}(u),d,\kappa,\mathrm{deps}(u),V\text{-log}\rangle$
\State \Return \texttt{CERTIFY} with $\pi_{\mathrm{der}}(u)$ \Comment{serve; cache only by policy}
\end{algorithmic}
\end{algorithm}

At time $t$, the core contains source facts $A_t$, source rules $\mathcal{R}_t$, cached consequences, and a certificate log. It starts empty or from a source-checked kernel with an empty cache. Every addition is checked. Erasure and eviction remove cached results that depend on a removed unit, so the state can both grow and shrink. After a kernel change, each retained cache certificate must be rechecked or rebuilt over the new kernel within depth $\delta$; entries that fail this check are evicted. Certificates record complete derivation DAGs and transitive source dependencies. Cache reuse is allowed only against the current kernel version. The persistent source store holds the kernel. Derived conclusions are recomputed or cached with their proof dependencies. The kernel result in Theorem~\ref{thm:TA} shows when this omission preserves \emph{closure}, the set of conclusions derivable from the source. The store--recompute threshold in Proposition~\ref{prop:scheduling} identifies units eligible for caching; the scheduler then selects a subset that fits the storage budget. Section~\ref{sec:theory} states the interface properties of this design formally (Theorem~\ref{thm:state-soundness}) and the assumptions under which they extend to coverage, storage, and reliability guarantees.
\subsection{Formal Guarantees and Operating Conditions}
\label{sec:theory}
Let $\SO$ denote the certified source set, $A$ its fixed-order irredundant kernel, and $\Cn$ closure in a decidable Horn fragment. The guarantee requires both valid inference and correct source interpretation. Theorem~\ref{thm:state-soundness} states the soundness assumptions on the two checkers. Table~\ref{tab:checker-contract} describes how the experimental source checkers implement these checks.
\begin{theorem}[Certified-state soundness]
\label{thm:state-soundness}
Assume that the source checker is sound for the declared grounding interface and that the fragment verifier is sound. Starting from the initialization in Section~\ref{sec:arch}, after any sequence of its transitions, (i) every persistent premise has a recheckable source-support record; (ii) every consequence $u$ served at time $t$ or cached in that state satisfies $u\in\Cn(A_t\cup\mathcal{R}_t)$ and $\Ddc(u\mid A_t\cup\mathcal{R}_t)\le\delta$; and (iii) erasing a source unit removes every cached consequence whose certificate depends on it.
\end{theorem}
The induction proof is in Appendix~\ref{app:proof-state}. Each admitted premise retains its evidence and checker version, and each derived unit retains a complete witness over the current kernel. The source-soundness assumption is not established for the experimental natural-language checkers, whose residual risk is measured empirically. The mathematical adapters execute recorded dependency graphs; they do not constitute an evaluation of end-to-end proof search.

\begin{proposition}[No promotion through restricted writes]
\label{prop:nonpromotion}
Fix a rule set $\mathcal R$ and a source vocabulary $\Omega$ disjoint from the predicates in rule heads, with all initial source facts in $\Omega$. Let the admission barrier reject any proposed source fact outside $\Omega$. Then adding an intermediate or target predicate to a proposal cannot admit that predicate as a source premise. Every accepted derivation of such a predicate must contain a rule application. This property holds for arbitrary model proposals and does not assume statistical independence or semantic accuracy of allowed leaf predicates.
\end{proposition}
\emph{Proof.} A rule-head predicate is outside $\Omega$, hence fails admission. It can enter closure only through the consequence of a rule. Rejected proposals contribute no premises. The same argument applies after every query-scoped reset. \hfill$\square$

Restricting source writes blocks direct insertion of the desired conclusion as a premise. Allowed leaf facts still require source checking. Appendix~\ref{app:gate-audit} examines these facts and tests the restriction by adding target conclusions to recorded proposals.
\begin{lemma}[Voting under shared bias and independent slips]
\label{lem:voting-floor}
Let $\bans$ be the query mass of a shared-error regime. Conditional on each query in this regime, suppose traces independently certify the same wrong answer with probability at least $1-\slipc$, where $\slipc<1/2$. On its complement, suppose they independently certify the unique correct answer with probability at least $1/2+\delta_0$, $\delta_0>0$. For $k$ traces under strict-majority serving (otherwise abstaining), the probability $\err(k)$ of serving a wrong answer obeys
\[
\bans\!\left(1-e^{-2k(1/2-\slipc)^2}\right)
\;\le\;\err(k)
\;\le\;\bans+e^{-2k\delta_0^2}.
\]
\end{lemma}
Appendix~\ref{app:proof-voting} gives the conditional bound. With two valid binary answers, majority voting with tie abstention is exactly the deployed agreement rule, and the fixed-trace source-recheck ablation in Section~\ref{subsec:rq2} measures what checking adds beyond it. For $k=2$ the bound can be loose; two calls alone neither identify its latent parameters nor establish conditional independence.

The deletion scan preserves $\Cn(A)=\Cn(\SO)$ (Theorem~\ref{thm:TA}), so the recomputed kernel decides every query that the certified source can decide, and any storage savings arise as a byproduct of whatever redundancy the scan finds in a given corpus.
\section{Experiments}
\label{sec:exp}
\begin{table}[!tbp]
\centering
\caption{Full-pool accuracy on public benchmarks (\%; abstentions count as errors).}
\label{tab:rq1-ledger}
\scriptsize
\resizebox{\columnwidth}{!}{%
\begin{tabular}{@{}lcccccc@{}}
\toprule
& \multicolumn{2}{c}{ContractNLI}
& \multicolumn{2}{c}{MedCalc-Bench Verified}
& \multicolumn{2}{c}{LeanDojo LD-MH} \\
\cmidrule(lr){2-3}\cmidrule(lr){4-5}\cmidrule(l){6-7}
Method & Seed & DeepSeek & Seed & DeepSeek & Seed & DeepSeek \\
\midrule
Pure CoT \citep{weietal2022cot} & $79.6$ & $79.9$
         & $66.1$ & $71.5$
         & $50.0$ & $51.5$ \\
RAG / open-book \citep{lewis2020rag} & $73.2$ & $74.7$
         & $76.2$ & $85.5$
         & $51.0$ & $50.0$ \\
\noalign{\color{cpunrow}\hrule height 2.4ex \vskip -2.4ex}
\textbf{CPUNeSy native} & $81.4$ & $70.9$ & $72.5$ & $63.5$
         & $\mathbf{92.5}$ & $79.0$ \\
\quad answered accuracy & $87.6$ & $88.7$ & $92.9$ & $95.2$ & $93.9$ & $88.3$ \\
\quad native coverage & $93.0$ & $79.9$ & $78.0$ & $66.6$ & $98.5$ & $89.5$ \\
\noalign{\color{cpunfallback}\hrule height 2.4ex \vskip -2.4ex}
\textbf{CPUNeSy + fallback} & $\mathbf{84.9}$ & $\mathbf{83.3}$
         & $\mathbf{90.1}$ & $\mathbf{90.4}$
         & $\mathbf{92.5}$ & $\mathbf{83.5}$ \\
\midrule
Logic-LM \citep{logiclm} & $0.0$ & $0.0$
         & $75.5$ & $77.5$
         & $50.5$ & $46.5$ \\
DetermLR \citep{sun2024determlr} & $48.4$ & $70.0$
         & $67.5$ & $61.5$
         & $50.5$ & $49.5$ \\
LINC \citep{olausson2023linc} & $16.0$ & $0.2$
     & $66.5$ & $44.7$
     & $49.5$ & $51.0$ \\
SymbCoT \citep{xu2024symbcot} & $56.8$ & $16.5$
        & $66.5$ & $75.1$
        & $49.5$ & $49.5$ \\
LINA \citep{li2024lina} & $48.5$ & $10.8$
     & $67.3$ & $69.3$
     & $49.5$ & $51.0$ \\
HBLR \citep{li2025hblr} & $54.7$ & $10.4$
     & $66.3$ & $62.8$
     & $49.0$ & $50.5$ \\
MenTaL \citep{mental} & $1.8$ & $0.0$
       & $75.0$ & $76.2$
       & $47.5$ & $50.0$ \\
\bottomrule
\end{tabular}}
\end{table}
\textbf{Setup.} Five suites span three domains. Legal reasoning uses L1-HARD ($240$ two- to four-hop queries with mid-link-break negatives) and ContractNLI; formal mathematics uses M-HARD ($200$ depth-$2$--$6$ AFP-DAG tasks) and LD-MH; medical calculation uses MedCalc-Bench Verified ($1{,}100$ items over $55$ verified calculators). Neural controls use chain-of-thought, complete-chain RAG, and high-reasoning-effort reruns at $2.2$--$3.1\times$ budget. Gradient-trained systems \citep{deepproblog,ltn} require labels absent from this serving-time setting, so the learned-component control is pretrained ReProver on LeanDojo Benchmark~4.
\textbf{Grounding models.} CPUNeSy uses K3 on controlled stress tests and the model named in each column on the public benchmarks.

\subsection{RQ1: Does Symbolic Derivation Recover Accuracy When Answers Must Be Derived?}
\label{subsec:rq1}

The public benchmarks separate two ways of serving (Table~\ref{tab:rq1-ledger}). The native certified channel already leads every baseline on LD-MH and on ContractNLI with Seed, and with the same-model fallback cascade it leads every baseline in every column. The gap between the native and fallback rows is itself the finding, and it is largest where the task asks the model to compose several rules or facts before answering.

On the L1-HARD mid-link-break negatives, chain-of-thought, retrieval, and Logic-LM all stay below eighty percent even when the reasoning budget is tripled. The certified channel serves essentially the entire pool at near-perfect accuracy (McNemar $p{<}10^{-5}$; full paired tests in Table~\ref{tab:paired-discordants}), and the twelve baseline failures we audit consistently omit exactly one intermediate link, which is what a mid-link-break negative is designed to expose.

The pattern repeats on the AFP dependency DAG. Chain-of-thought and retrieval hover around chance because the answer is only well-defined relative to a chain of prior lemmas, while CPUNeSy reconstructs that chain from the stored source and serves the majority of queries at high accuracy.

MedCalc shows a different split that clarifies what execution contributes. The certified channel is highly precise on the questions it answers, but the medical rubric forces frequent abstention when extracted parameters disagree. Routing those abstentions to the same model's open-book answer lifts full-pool accuracy from $72.5\%$ to $90.1\%$ for Seed and from $63.5\%$ to $90.4\%$ for DeepSeek without touching the certified subset, and analogous same-model cascades give a milder but consistent benefit on ContractNLI and LD-MH.

On LeanDojo Benchmark~4, the kernel pool and BM25 both attain seventy-six percent full gold-premise coverage, roughly an order of magnitude above the learned ReProver retriever \citep{leandojo}. The kernel pool therefore functions as a competitive premise source in its own right. Appendix~\ref{app:exp} gives full protocol details, exact counts, and provenance.

\begin{figure}[!tbp]
\centering
\includegraphics[width=0.9\columnwidth]{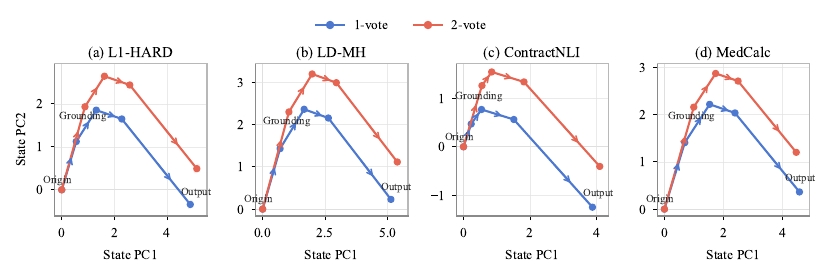}
\caption{Logged process-state trajectories. Paths join origin, grounding, admission, closure, and output; the first two principal components explain most of the variance in the pooled log features, and ellipses are bootstrap $95\%$ centroid regions. Models and runs are listed in Appendix~\ref{app:figure-provenance}.}
\label{fig:reasoning-state-trajectories}
\end{figure}

\subsection{RQ2: Which Component Recovers Accuracy, and What Does Certification Add?}
\label{subsec:rq2}

\begin{wraptable}[9]{r}{0.36\columnwidth}
\vspace{-\baselineskip}
\centering
\caption{L1-HARD ablation (\%).}
\label{tab:rq2-ablation}
\small
\begin{tabular}{@{}lrr@{}}
\toprule
Method & Full & Break \\
\midrule
Pure & $84.6$ & $69.2$ \\
RAG & $87.9$ & $75.8$ \\
Logic-LM & $83.3$ & $68.3$ \\
Solver & $\mathbf{98.8}$ & $\mathbf{99.2}$ \\
\textbf{CPUNeSy} & $97.1$ & $98.3$ \\
\bottomrule
\end{tabular}
\end{wraptable}
RQ2 separates deterministic execution from the additional controls on model writes and serving. Execution accounts for most of the accuracy recovery on the controlled derivation tasks, while the added controls change which answers are served, trading coverage for lower accepted-answer risk. On L1-HARD, the uncertified one-vote solver reaches $98.8\%$ full-pool accuracy and the restricted two-vote policy $97.1\%$ (Table~\ref{tab:rq2-ablation}). Execution alone therefore recovers most of the gain over neural baselines, and the paired difference is not statistically significant ($p=0.34$; the interval does not establish equivalence either, Appendix~\ref{app:gate-audit}). We therefore treat certification as a selective-serving control rather than an accuracy mechanism. Its value is the error--abstention tradeoff it imposes on top of execution and voting. Because this comparison changes the number of grounding calls, the fixed-trace experiment below isolates the additional effect of source rechecking.

We hold both ContractNLI grounding traces fixed and compare plain two-vote agreement with agreement plus source rechecking. For DeepSeek, rechecking withholds $63/252=25\%$ of the wrong answers that survive agreement and $51$ correct answers. For Seed, it withholds $8/128=6.25\%$ of such wrong answers and $14$ correct answers. This isolates an error--abstention tradeoff beyond voting. Separately, the restricted L1-HARD interface has no observed wrong answers among $233$ served queries; its historical runs and the explicit admission replay have different implementation scopes (Table~\ref{tab:disposition}). The logged state differences in Figures~\ref{fig:reasoning-state-trajectories} and~\ref{fig:rq2-mechanism-evidence} are descriptive, and the attribution to rechecking rests on the paired outcomes rather than those projections. Appendix~\ref{app:gate-audit} reports the complete counts.

\begin{center}
\vspace{-4pt}
\includegraphics[width=0.9\columnwidth]{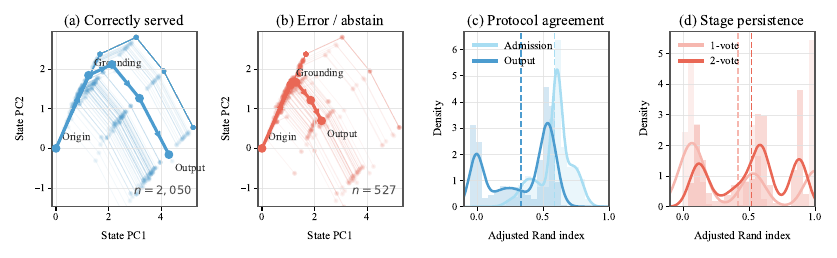}\par
\vspace{-4pt}
\end{center}
\vspace{-6pt}
\noindent{\footnotesize\makeatletter\refstepcounter{figure}\label{fig:rq2-mechanism-evidence}Figure~\thefigure: State-cluster consistency across domains. (a--b) Two-vote trajectories for correctly served versus other queries. (c--d) Adjusted Rand index (ARI), with paired bootstrap intervals, measures agreement between protocol clusters and their persistence across stages.\makeatother\par}
\vspace{-2pt}

\subsection{RQ3: Which Observed Conditions Explain Utility?}
\label{subsec:rq3}

\begin{wraptable}[10]{r}{0.55\columnwidth}
\vspace{-\baselineskip}
\centering
\caption{CPUNeSy by grounding regime (L1-HARD, \%).}
\label{tab:cacc3}
\footnotesize
\setlength{\tabcolsep}{2.4pt}
\begin{tabular}{@{}llrrrr@{}}
\toprule
Grounder & regime & full & cov. & neural & $\Delta_{\rm full}$ \\
\midrule
K3 & slip & $97.1$ & $97.1$ & $87.9$ & $+9.2$ \\
K2.6 & slip & $94.2$ & $94.6$ & $88.3$ & $+5.9$ \\
GLM & slip & $92.9$ & $93.8$ & $92.9$ & $0.0$ \\
\midrule
DeepSeek & bias & $76.7$ & $91.7$ & $84.6$ & $-7.9$ \\
Qwen2.5-72B & bias & $82.1$ & $93.3$ & $90.0$ & $-7.9$ \\
\bottomrule
\end{tabular}
\end{wraptable}
RQ1 and RQ2 motivate three diagnostics for interpreting utility. \emph{Derivation demand} counts required rule applications; \emph{kernel adequacy} asks whether the stored source supports a query; and the \emph{grounding-error pattern} describes errors shared across calls versus disagreements. We analyze these properties retrospectively on the recorded runs. The large gains on L1-HARD, M-HARD, and LD-MH are consistent with the role of derivation and source coverage, while the public-benchmark results show that checks can reduce wrong answers without improving full-pool accuracy.

Grouping the L1-HARD runs by their observed joint-error patterns gives a clean split (Table~\ref{tab:cacc3}). Relative to each model's better neural control, full-pool accuracy increases for K3 and K2.6, is unchanged for GLM, and decreases for DeepSeek and Qwen. This pattern is consistent with the shared-error limitation in Lemma~\ref{lem:voting-floor}; it does not identify the lemma's latent parameters or establish independence from two calls. Agreement can retain shared errors, which is what the fixed-trace source-recheck comparison in RQ2 measures. Table~\ref{tab:crossckpt} gives the full cross-model results. Correct answers, errors, and abstentions also distribute differently over vote disagreement and grounding coverage (Figure~\ref{fig:grounding-geometry}), and the overlap on ContractNLI illustrates why these two coordinates alone do not establish source support. These observations motivate a calibration protocol. 

\begin{center}
\includegraphics[width=0.9\columnwidth]{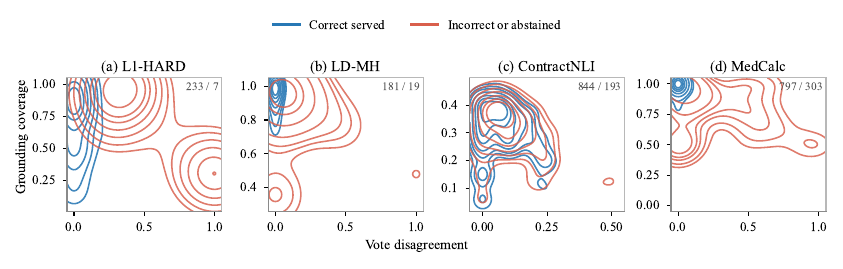}
\end{center}
\noindent{\footnotesize\makeatletter\refstepcounter{figure}\label{fig:grounding-geometry}Figure~\thefigure: Vote-level grounding geometry (K3 for L1-HARD/LD-MH; Seed for ContractNLI/MedCalc). The LD-MH panel uses a different model from the Seed column in Table~\ref{tab:rq1-ledger}. The horizontal axis shows vote disagreement (one minus Jaccard) and the vertical axis grounding coverage. Blue marks correctly answered queries and coral marks errors or abstentions; contours are normalized within each class.\makeatother\par}
\vspace{-2pt}

\subsection{RQ4: What Is the Serving-Time Cost of Controlling Model Writes?}
\label{subsec:rq4}

The resource ledger distinguishes candidate-pool restriction, closure-preserving kernelization, and selective serving. On M-HARD, restricting the candidate vocabulary to leaf-$\Omega$ retains $48.9\%$ of full-pool units at the same observed $97.45\%$ answered accuracy and reduces tokens by $43\%$ (Table~\ref{tab:disposition}). These are candidate-restriction measurements rather than a measured storage-byte saving from the deletion scan. The separate corpus audit finds bounded rederivation redundancy of $37.1\%$ for AFP and $0.60\%$ for law, both below the protocol's $\tau=0.5$ compression-yield target (Table~\ref{tab:kstep}). Physical storage and the scheduler-specific crossover remain unmeasured.

Repeated grounding and rechecking trade wrong answers against abstentions and compute, and their utility depends on all three costs. Under the measured token ledger, restricted two-call serving on L1-HARD is preferred when $2c_W>6c_A+145{,}422c_T$ over the $240$ items, with $c_W$, $c_A$, and $c_T$ pricing a wrong answer, an abstention, and a token (derivation in Appendix~\ref{app:gate-audit}). Counts and answered-risk intervals for each recorded policy, including fallback, appear in Table~\ref{tab:selective-points}.
\section{Conclusion}
\label{sec:conclusion}
When a language model is paired with a symbolic executor, two things change at once, and they are easy to confuse. The model stops deriving by sampling and starts deriving by solving, and its premises become gated writes instead of free text. \textsc{CPUNeSy} separates the two, pairing a certificate gate that controls what the model may write into symbolic state with a deterministic executor that derives the answer. Execution accounts for most of the accuracy recovery on derivation-heavy tasks, lifting full-pool accuracy on L1-HARD from $87.9\%$ for the strongest neural baseline to $98.8\%$ with no certificate at all. The certificate adds something different. With grounding traces held fixed on ContractNLI, source rechecking removes a quarter of DeepSeek's wrong answers that survive two-vote agreement, at the cost of withholding correct answers it cannot support. The practical reading is that a deployment needing accuracy on derivation-heavy tasks gets most of it from the solver alone, while one whose failures must not be silent needs the gate, at roughly twice the tokens. The boundaries are equally specific. \textsc{CPUNeSy}'s certificates cover derivational validity relative to admitted premises, so strengthening the method reduces to strengthening the source checker, and where grounding errors are shared across calls, agreement repeats them and the gate cannot help. The next step is prospective calibration, measuring derivation demand, kernel adequacy, and grounding regime on a disjoint annotated set before freezing the serving policy.
\section*{Ethics Statement}
The legal and medical experiments evaluate benchmark inference and calculation. Clinical or legal use would require further validation of source interpretation, since the model can misread a contract or clinical note even when the subsequent inference passes verification. Applications should expose source evidence, checker identity, and abstention status, and distinguish fallback answers from checked outputs. Decisions affecting patients or legal rights require appropriate expert review.

\section*{AI Use Statement}
Language models served as experimental subjects and grounding components, and assisted coding, writing, visualization, and offline analysis. Human authors are responsible for verifying all records, results, and claims.

\begingroup
\raggedright
\Urlmuskip=0mu plus 2mu\relax
\bibliography{iclr2027_conference}

@article{oit,
  author    = {Qiu, Siyuan and Xu, Jianfeng},
  title     = {Research on a General State Formalization Method from the Perspective of Logic},
  journal   = {Mathematics},
  volume    = {13},
  number    = {20},
  pages     = {3324},
  year      = {2025},
  publisher = {MDPI},
  doi       = {10.3390/math13203324}
}

@article{semrd,
  author  = {Xu, Jianfeng},
  title   = {Rate-Distortion Theory for Deductive Sources under Closure Fidelity},
  journal = {arXiv preprint arXiv:2604.15698},
  year    = {2026},
  doi     = {10.48550/arXiv.2604.15698},
  url     = {https://arxiv.org/abs/2604.15698}
}

@article{derdepth,
  author  = {Xu, Jianfeng},
  title   = {Derivation Depth as an Information Metric: Axioms, Coding Theorems, and Storage--Computation Tradeoffs},
  journal = {arXiv preprint arXiv:2602.19137},
  year    = {2026},
  doi     = {10.48550/arXiv.2602.19137},
  url     = {https://arxiv.org/abs/2602.19137}
}

@article{pvcache,
  author  = {Xu, Jianfeng},
  title   = {Proof-Valid Caching under Premise Erasures: Local Structural Limits and Shared-Workload Gains},
  journal = {arXiv preprint arXiv:2608.11782},
  year    = {2026},
  doi     = {10.48550/arXiv.2608.11782},
  url     = {https://arxiv.org/abs/2608.11782}
}

@article{pvbench,
  author  = {Xu, Jianfeng},
  title   = {Proof-Valid Benchmarking for Approximate {LLM} Caching under Premise Erasures},
  journal = {ChinaXiv preprint ChinaXiv:202608.00158},
  year    = {2026},
  doi     = {10.12074/202608.00158},
  url     = {https://chinaxiv.org/abs/202608.00158}
}

@inproceedings{deepproblog,
  author    = {Manhaeve, Robin and Duman{\v{c}}i{\'c}, Sebastijan and Kimmig, Angelika and Demeester, Thomas and De Raedt, Luc},
  title     = {{DeepProbLog}: Neural Probabilistic Logic Programming},
  booktitle = {Advances in Neural Information Processing Systems},
  volume    = {31},
  year      = {2018},
  publisher = {Curran Associates, Inc.},
  url       = {https://proceedings.neurips.cc/paper_files/paper/2018/hash/dc5d637ed5e62c36ecb73b654b05ba2a-Abstract.html}
}

@inproceedings{neurasp,
  author    = {Yang, Zhun and Ishay, Adam and Lee, Joohyung},
  title     = {{NeurASP}: Embracing Neural Networks into Answer Set Programming},
  booktitle = {Proceedings of the Twenty-Ninth International Joint Conference on Artificial Intelligence (IJCAI)},
  pages     = {1755--1762},
  year      = {2020},
  doi       = {10.24963/ijcai.2020/243}
}

@inproceedings{logiclm,
  author    = {Pan, Liangming and Albalak, Alon and Wang, Xinyi and Wang, William Yang},
  title     = {{Logic-LM}: Empowering Large Language Models with Symbolic Solvers for Faithful Logical Reasoning},
  booktitle = {Findings of the Association for Computational Linguistics: EMNLP 2023},
  pages     = {3806--3824},
  year      = {2023},
  doi       = {10.18653/v1/2023.findings-emnlp.248}
}

@article{lnn,
  author  = {Riegel, Ryan and Gray, Alexander and Luus, Francois and Khan, Naweed and Makondo, Ndivhuwo and Akhalwaya, Ismail Yunus and Qian, Haifeng and Fagin, Ronald and Barahona, Francisco and Sharma, Udit and Ikbal, Shajith and Karanam, Hima and Neelam, Sumit and Likhyani, Ankita and Srivastava, Santosh},
  title   = {Logical Neural Networks},
  journal = {arXiv preprint arXiv:2006.13155},
  year    = {2020},
  doi     = {10.48550/arXiv.2006.13155},
  url     = {https://arxiv.org/abs/2006.13155}
}

@article{ltn,
  author  = {Badreddine, Samy and d'Avila Garcez, Artur and Serafini, Luciano and Spranger, Michael},
  title   = {Logic Tensor Networks},
  journal = {Artificial Intelligence},
  volume  = {303},
  pages   = {103649},
  year    = {2022},
  doi     = {10.1016/j.artint.2021.103649}
}

@inproceedings{satnet,
  author    = {Wang, Po-Wei and Donti, Priya L. and Wilder, Bryan and Kolter, J. Zico},
  title     = {{SATNet}: Bridging Deep Learning and Logical Reasoning Using a Differentiable Satisfiability Solver},
  booktitle = {International Conference on Machine Learning (ICML)},
  pages     = {6545--6554},
  year      = {2019}
}

@article{tensorlog,
  author  = {Cohen, William W. and Yang, Fan and Mazaitis, Kathryn Rivard},
  title   = {{TensorLog}: A Probabilistic Database Implemented Using Deep-Learning Infrastructure},
  journal = {Journal of Artificial Intelligence Research},
  volume  = {67},
  pages   = {285--325},
  year    = {2020},
  doi     = {10.1613/jair.1.11944}
}

@inproceedings{e2eprover,
  author    = {Rockt{\"a}schel, Tim and Riedel, Sebastian},
  title     = {End-to-End Differentiable Proving},
  booktitle = {Advances in Neural Information Processing Systems},
  volume    = {30},
  year      = {2017},
  publisher = {Curran Associates, Inc.},
  url       = {https://proceedings.neurips.cc/paper_files/paper/2017/hash/b2ab001909a8a6f04b51920306046ce5-Abstract.html}
}

@inproceedings{slash,
  author    = {Skryagin, Arseny and Stammer, Wolfgang and Ochs, Daniel and Dhami, Devendra Singh and Kersting, Kristian},
  title     = {Neural-Probabilistic Answer Set Programming},
  booktitle = {Proceedings of the International Conference on Principles of Knowledge Representation and Reasoning (KR)},
  pages     = {463--473},
  year      = {2022},
  doi       = {10.24963/kr.2022/48}
}

@inproceedings{shortcuts,
  author    = {Marconato, Emanuele and Teso, Stefano and Vergari, Antonio and Passerini, Andrea},
  title     = {Not All Neuro-Symbolic Concepts Are Created Equal: Analysis and Mitigation of Reasoning Shortcuts},
  booktitle = {Advances in Neural Information Processing Systems (NeurIPS)},
  pages     = {72507--72539},
  year      = {2023},
  doi       = {10.52202/075280-3170}
}

@inproceedings{magnushammer,
  author    = {Miku{\l}a, Maciej and Tworkowski, Szymon and Antoniak, Szymon and Piotrowski, Bartosz and Jiang, Albert Qiaochu and Zhou, Jin Peng and Szegedy, Christian and Kuci{\'n}ski, {\L}ukasz and Mi{\l}o{\'s}, Piotr and Wu, Yuhuai},
  title     = {{MagnusHammer}: A Transformer-Based Approach to Premise Selection},
  booktitle = {International Conference on Learning Representations (ICLR)},
  year      = {2024}
}

@article{dl4tp,
  author  = {Li, Zhaoyu and Sun, Jialiang and Murphy, Logan and Su, Qidong and Li, Zenan and Zhang, Xian and Yang, Kaiyu and Si, Xujie},
  title   = {A Survey on Deep Learning for Theorem Proving},
  journal = {arXiv preprint arXiv:2404.09939},
  year    = {2024},
  doi     = {10.48550/arXiv.2404.09939},
  url     = {https://arxiv.org/abs/2404.09939}
}

@article{nesy2024,
  author  = {Colelough, Brandon C. and Regli, William},
  title   = {Neuro-Symbolic {AI} in 2024: A Systematic Review},
  journal = {arXiv preprint arXiv:2501.05435},
  year    = {2025}
}

@article{hoeffding,
  author  = {Hoeffding, Wassily},
  title   = {Probability Inequalities for Sums of Bounded Random Variables},
  journal = {Journal of the American Statistical Association},
  volume  = {58},
  number  = {301},
  pages   = {13--30},
  year    = {1963}
}

@incollection{garcez2022,
  author    = {Besold, Tarek R. and d'Avila Garcez, Artur and Bader, Sebastian and Bowman, Howard and Domingos, Pedro and Hitzler, Pascal and others},
  title     = {Neural-Symbolic Learning and Reasoning: A Survey and Interpretation},
  booktitle = {Neuro-Symbolic Artificial Intelligence: The State of the Art},
  publisher = {IOS Press},
  pages     = {1--51},
  year      = {2022}
}

@inproceedings{souffle,
  author    = {Scholz, Bernhard and Jordan, Herbert and Suboti{\'c}, Pavle and Westmann, Till},
  title     = {On Fast Large-Scale Program Analysis in {Datalog}},
  booktitle = {Proceedings of the International Conference on Compiler Construction (CC)},
  pages     = {196--206},
  year      = {2016},
  doi       = {10.1145/2892208.2892226}
}

@inproceedings{dsp,
  author    = {Jiang, Albert Q. and Welleck, Sean and Zhou, Jin Peng and Lacroix, Timoth{\'e}e and Liu, Jiacheng and Li, Wenda and Jamnik, Mateja and Lample, Guillaume and Wu, Yuhuai},
  title     = {Draft, Sketch, and Prove: Guiding Formal Theorem Provers with Informal Proofs},
  booktitle = {International Conference on Learning Representations (ICLR)},
  year      = {2023}
}

@article{kbann,
  author  = {Towell, Geoffrey G. and Shavlik, Jude W.},
  title   = {Knowledge-Based Artificial Neural Networks},
  journal = {Artificial Intelligence},
  volume  = {70},
  number  = {1--2},
  pages   = {119--165},
  year    = {1994}
}

@article{cilp,
  author  = {d'Avila Garcez, Artur S. and Zaverucha, Gerson},
  title   = {The Connectionist Inductive Learning and Logic Programming System},
  journal = {Applied Intelligence},
  volume  = {11},
  number  = {1},
  pages   = {59--77},
  year    = {1999}
}

@article{mln,
  author  = {Richardson, Matthew and Domingos, Pedro},
  title   = {Markov Logic Networks},
  journal = {Machine Learning},
  volume  = {62},
  number  = {1--2},
  pages   = {107--136},
  year    = {2006}
}

@inproceedings{problog,
  author    = {De Raedt, Luc and Kimmig, Angelika and Toivonen, Hannu},
  title     = {{ProbLog}: A Probabilistic {Prolog} and Its Application in Link Discovery},
  booktitle = {Proceedings of the Twentieth International Joint Conference on Artificial Intelligence (IJCAI)},
  pages     = {2462--2467},
  year      = {2007}
}

@inproceedings{semloss,
  author    = {Xu, Jingyi and Zhang, Zilu and Friedman, Tal and Liang, Yitao and Van den Broeck, Guy},
  title     = {A Semantic Loss Function for Deep Learning with Symbolic Knowledge},
  booktitle = {International Conference on Machine Learning (ICML)},
  pages     = {5502--5511},
  year      = {2018}
}

@article{dilp,
  author  = {Evans, Richard and Grefenstette, Edward},
  title   = {Learning Explanatory Rules from Noisy Data},
  journal = {Journal of Artificial Intelligence Research (JAIR)},
  volume  = {61},
  pages   = {1--64},
  year    = {2018}
}

@article{gptf,
  author  = {Polu, Stanislas and Sutskever, Ilya},
  title   = {Generative Language Modeling for Automated Theorem Proving},
  journal = {arXiv preprint arXiv:2009.03393},
  year    = {2020},
  doi     = {10.48550/arXiv.2009.03393},
  url     = {https://arxiv.org/abs/2009.03393}
}

@inproceedings{htps,
  author    = {Lample, Guillaume and Lacroix, Timoth{\'e}e and Lachaux, Marie-Anne and Rodriguez, Aur{\'e}lien and Hayat, Amaury and Lavril, Thibaut and Ebner, Gabriel and Martinet, Xavier},
  title     = {{HyperTree} Proof Search for Neural Theorem Proving},
  booktitle = {Advances in Neural Information Processing Systems (NeurIPS)},
  pages     = {26337--26349},
  year      = {2022}
}

@inproceedings{leandojo,
  author    = {Yang, Kaiyu and Swope, Aidan and Gu, Alex and Chalamala, Rahul and Song, Peiyang and Yu, Shixing and Godil, Saad and Prenger, Ryan and Anandkumar, Anima},
  title     = {{LeanDojo}: Theorem Proving with Retrieval-Augmented Language Models},
  booktitle = {Advances in Neural Information Processing Systems (NeurIPS), Datasets and Benchmarks Track},
  year      = {2023}
}

@inproceedings{contractnli,
  author    = {Koreeda, Yuta and Manning, Christopher},
  title     = {{ContractNLI}: A Dataset for Document-level Natural Language Inference for Contracts},
  booktitle = {Findings of the Association for Computational Linguistics: EMNLP 2021},
  pages     = {1907--1919},
  year      = {2021}
}

@article{alphageometry,
  author  = {Trinh, Trieu H. and Wu, Yuhuai and Le, Quoc V. and He, He and Luong, Thang},
  title   = {Solving Olympiad Geometry without Human Demonstrations},
  journal = {Nature},
  volume  = {625},
  number  = {7995},
  pages   = {476--482},
  year    = {2024}
}

@inproceedings{lewis2020rag,
  author    = {Lewis, Patrick and Perez, Ethan and Piktus, Aleksandra and Petroni, Fabio and Karpukhin, Vladimir and Goyal, Naman and K{\"u}ttler, Heinrich and Lewis, Mike and Yih, Wen-tau and Rockt{\"a}schel, Tim and Riedel, Sebastian and Kiela, Douwe},
  title     = {Retrieval-Augmented Generation for Knowledge-Intensive {NLP} Tasks},
  booktitle = {Advances in Neural Information Processing Systems (NeurIPS)},
  pages     = {9459--9474},
  year      = {2020}
}

@inproceedings{wang2023selfconsistency,
  author    = {Wang, Xuezhi and Wei, Jason and Schuurmans, Dale and Le, Quoc V. and Chi, Ed H. and Narang, Sharan and Chowdhery, Aakanksha and Zhou, Denny},
  title     = {Self-Consistency Improves Chain of Thought Reasoning in Language Models},
  booktitle = {International Conference on Learning Representations (ICLR)},
  year      = {2023}
}

@inproceedings{gao2023pal,
  author    = {Gao, Luyu and Madaan, Aman and Zhou, Shuyan and Alon, Uri and Liu, Pengfei and Yang, Yiming and Callan, Jamie and Neubig, Graham},
  title     = {{PAL}: Program-aided Language Models},
  booktitle = {International Conference on Machine Learning (ICML)},
  pages     = {10764--10799},
  year      = {2023}
}

@inproceedings{geifman2017selective,
  author    = {Geifman, Yonatan and El-Yaniv, Ran},
  title     = {Selective Classification for Deep Neural Networks},
  booktitle = {Advances in Neural Information Processing Systems (NeurIPS)},
  pages     = {4878--4887},
  year      = {2017}
}

@article{cobbe2021verifier,
  author  = {Cobbe, Karl and Kosaraju, Vineet and Bavarian, Mohammad and Chen, Mark and Jun, Heewoo and Kaiser, Lukasz and Plappert, Matthias and Tworek, Jerry and Hilton, Jacob and Nakano, Reiichiro and Hesse, Christopher and Schulman, John},
  title   = {Training Verifiers to Solve Math Word Problems},
  journal = {arXiv preprint arXiv:2110.14168},
  year    = {2021},
  doi     = {10.48550/arXiv.2110.14168},
  url     = {https://arxiv.org/abs/2110.14168}
}

@inproceedings{helff2026slr,
  author    = {Helff, Lukas and Omar, Ahmad and Friedrich, Felix and W{\"u}st, Antonia and Shindo, Hikaru and Mitchell, Rupert and Woydt, Tim and Schramowski, Patrick and Stammer, Wolfgang and Kersting, Kristian},
  title     = {{SLR}: Automated Synthesis for Scalable Logical Reasoning},
  booktitle = {Proceedings of the 64th Annual Meeting of the Association for Computational Linguistics (Volume 1: Long Papers)},
  pages     = {402--426},
  year      = {2026},
  publisher = {Association for Computational Linguistics},
  doi       = {10.18653/v1/2026.acl-long.16},
  url       = {https://aclanthology.org/2026.acl-long.16/}
}

@inproceedings{liu2025safe,
  author    = {Liu, Chengwu and Yuan, Ye and Yin, Yichun and Xu, Yan and Xu, Xin and Chen, Zaoyu and Wang, Yasheng and Shang, Lifeng and Liu, Qun and Zhang, Ming},
  title     = {{Safe}: Enhancing Mathematical Reasoning in Large Language Models via Retrospective Step-aware Formal Verification},
  booktitle = {Proceedings of the 63rd Annual Meeting of the Association for Computational Linguistics (ACL)},
  pages     = {12171--12186},
  year      = {2025},
  doi       = {10.18653/v1/2025.acl-long.594},
  url       = {https://aclanthology.org/2025.acl-long.594/}
}

@article{galitsky2026metapcr,
  author  = {Galitsky, Boris and Solodkin, Vladimir and Beznosikov, Aleksandr},
  title   = {{MetaPCR-LLM}: Universal Proof-Carrying Reasoning Across Heterogeneous Logics for Large Language Model Validation},
  journal = {Preprints.org preprint 202609.0849},
  year    = {2026},
  doi     = {10.20944/preprints202609.0849.v1},
  url     = {https://www.preprints.org/manuscript/202609.0849}
}

@inproceedings{kordjamshidi2026reasoners,
  author    = {Kordjamshidi, Parisa and Aslan, Samer and Seshadri, Madhavan and Barrett, Leslie and Santus, Enrico},
  title     = {Reasoners or Translators? Contamination-aware Evaluation and Neuro-Symbolic Robustness on Tax Law},
  booktitle = {Proceedings of the First Workshop on Structured Understanding, Retrieval, and Generation in the {LLM} Era ({SURG}e{LLM} 2026)},
  pages     = {344--360},
  year      = {2026},
  publisher = {Association for Computational Linguistics},
  doi       = {10.18653/v1/2026.surgellm-1.23},
  url       = {https://aclanthology.org/2026.surgellm-1.23/}
}

@article{wu2026logicgraph,
  author  = {Wu, Yanrui and Zhang, Lingling and Zhang, Xinyu and Chang, Jiayu and Li, Pengyu and Jiang, Xu and Hu, Jingtao and Liu, Jun},
  title   = {{LogicGraph}: Benchmarking Multi-Path Logical Reasoning via Neuro-Symbolic Generation and Verification},
  journal = {arXiv preprint arXiv:2602.21044},
  year    = {2026},
  doi     = {10.48550/arXiv.2602.21044},
  url     = {https://arxiv.org/abs/2602.21044}
}

@inproceedings{weietal2022cot,
  author    = {Wei, Jason and Wang, Xuezhi and Schuurmans, Dale and Bosma, Maarten and Ichter, Brian and Xia, Fei and Chi, Ed H. and Le, Quoc V. and Zhou, Denny},
  title     = {Chain-of-Thought Prompting Elicits Reasoning in Large Language Models},
  booktitle = {Advances in Neural Information Processing Systems (NeurIPS)},
  pages     = {24824--24837},
  year      = {2022}
}

@article{kaplan2020scaling,
  author  = {Kaplan, Jared and McCandlish, Sam and Henighan, Tom and Brown, Tom B. and Chess, Benjamin and Child, Rewon and Gray, Scott and Radford, Alec and Wu, Jeffrey and Amodei, Dario},
  title   = {Scaling Laws for Neural Language Models},
  journal = {arXiv preprint arXiv:2001.08361},
  year    = {2020},
  doi     = {10.48550/arXiv.2001.08361},
  url     = {https://arxiv.org/abs/2001.08361}
}

@article{srivastava2022bigbench,
  author  = {Srivastava, Aarohi and Rastogi, Abhinav and Rao, Abhishek and Shoeb, Abu Awal Md and Abid, Abubakar and Fisch, Adam and Brown, Adam R. and Santoro, Adam and Gupta, Aditya and Garriga-Alonso, Adri{\`a} and others},
  title   = {Beyond the Imitation Game: Quantifying and Extrapolating the Capabilities of Language Models},
  journal = {Transactions on Machine Learning Research (TMLR)},
  year    = {2023}
}

@inproceedings{dziri2024faith,
  author    = {Dziri, Nouha and Lu, Ximing and Sclar, Melanie and Li, Xiang Lorraine and Jiang, Liwei and Lin, Bill Yuchen and Welleck, Sean and West, Peter and Bhagavatula, Chandra and Le Bras, Ronan and Hwang, Jena D. and Sanyal, Soumya and Ren, Xiang and Ettinger, Allyson and Harchaoui, Zaid and Choi, Yejin},
  title     = {Faith and Fate: Limits of Transformers on Compositionality},
  booktitle = {Advances in Neural Information Processing Systems (NeurIPS)},
  pages     = {70293--70332},
  year      = {2023}
}

@inproceedings{thakur2023copra,
  author    = {Thakur, Amitayush and Tsoukalas, George and Wen, Yeming and Xin, Jimmy and Chaudhuri, Swarat},
  title     = {An In-Context Learning Agent for Formal Theorem-Proving},
  booktitle = {First Conference on Language Modeling (COLM)},
  year      = {2024},
  note      = {arXiv:2310.04353}
}

@inproceedings{olausson2023linc,
  author    = {Olausson, Theo and Gu, Alex and Lipkin, Ben and Zhang, Cedegao and Solar-Lezama, Armando and Tenenbaum, Joshua and Levy, Roger},
  title     = {{LINC}: A Neurosymbolic Approach for Logical Reasoning by Combining Language Models with First-Order Logic Provers},
  booktitle = {Proceedings of the 2023 Conference on Empirical Methods in Natural Language Processing},
  pages     = {5153--5176},
  year      = {2023},
  publisher = {Association for Computational Linguistics},
  doi       = {10.18653/v1/2023.emnlp-main.313},
  url       = {https://aclanthology.org/2023.emnlp-main.313/}
}

@inproceedings{ye2023satlm,
  author    = {Ye, Xi and Chen, Qiaochu and Dillig, Isil and Durrett, Greg},
  title     = {{SatLM}: Satisfiability-Aided Language Models Using Declarative Prompting},
  booktitle = {Advances in Neural Information Processing Systems},
  volume    = {36},
  pages     = {45548--45580},
  year      = {2023},
  publisher = {Curran Associates, Inc.},
  doi       = {10.52202/075280-1974},
  url       = {https://proceedings.neurips.cc/paper_files/paper/2023/file/8e9c7d4a48bdac81a58f983a64aaf42b-Paper-Conference.pdf}
}

@inproceedings{xu2024symbcot,
  author    = {Xu, Jundong and Fei, Hao and Pan, Liangming and Liu, Qian and Lee, Mong-Li and Hsu, Wynne},
  title     = {Faithful Logical Reasoning via Symbolic Chain-of-Thought},
  booktitle = {Proceedings of the 62nd Annual Meeting of the Association for Computational Linguistics (Volume 1: Long Papers)},
  pages     = {13326--13365},
  year      = {2024},
  publisher = {Association for Computational Linguistics},
  doi       = {10.18653/v1/2024.acl-long.720},
  url       = {https://aclanthology.org/2024.acl-long.720/}
}

@inproceedings{jiang2024leanreasoner,
  author    = {Jiang, Dongwei and Fonseca, Marcio and Cohen, Shay B.},
  title     = {{LeanReasoner}: Boosting Complex Logical Reasoning with Lean},
  booktitle = {Proceedings of the 2024 Conference of the North American Chapter of the Association for Computational Linguistics: Human Language Technologies (Volume 1: Long Papers)},
  pages     = {7497--7510},
  year      = {2024},
  publisher = {Association for Computational Linguistics},
  doi       = {10.18653/v1/2024.naacl-long.416},
  url       = {https://aclanthology.org/2024.naacl-long.416/}
}

@inproceedings{ryu2025clover,
  author    = {Ryu, Hyun and Kim, Gyeongman and Lee, Hyemin S. and Yang, Eunho},
  title     = {Divide and Translate: Compositional First-Order Logic Translation and Verification for Complex Logical Reasoning},
  booktitle = {International Conference on Learning Representations},
  year      = {2025},
  url       = {https://openreview.net/forum?id=09FiNmvNMw}
}

@inproceedings{xu2025aristotle,
  author    = {Xu, Jundong and Fei, Hao and Luo, Meng and Liu, Qian and Pan, Liangming and Wang, William Yang and Nakov, Preslav and Lee, Mong-Li and Hsu, Wynne},
  title     = {Aristotle: Mastering Logical Reasoning with A Logic-Complete Decompose-Search-Resolve Framework},
  booktitle = {Proceedings of the 63rd Annual Meeting of the Association for Computational Linguistics (Volume 1: Long Papers)},
  pages     = {3052--3075},
  year      = {2025},
  publisher = {Association for Computational Linguistics},
  doi       = {10.18653/v1/2025.acl-long.153},
  url       = {https://aclanthology.org/2025.acl-long.153/}
}

@inproceedings{yang2025caring,
  author    = {Yang, Sen and Li, Xin and Cui, Leyang and Bing, Lidong and Lam, Wai},
  title     = {Neuro-Symbolic Integration Brings Causal and Reliable Reasoning Proofs},
  booktitle = {Findings of the Association for Computational Linguistics: NAACL 2025},
  pages     = {5732--5744},
  year      = {2025},
  publisher = {Association for Computational Linguistics},
  doi       = {10.18653/v1/2025.findings-naacl.317},
  url       = {https://aclanthology.org/2025.findings-naacl.317/}
}

@inproceedings{qi2025provergen,
  author    = {Qi, Chengwen and Ma, Ren and Li, Bowen and Du, He and Hui, Binyuan and Wu, Jinwang and Laili, Yuanjun and He, Conghui},
  title     = {Large Language Models Meet Symbolic Provers for Logical Reasoning Evaluation},
  booktitle = {International Conference on Learning Representations},
  year      = {2025},
  url       = {https://openreview.net/forum?id=C25SgeXWjE}
}

@inproceedings{parmar2024logicbench,
  author    = {Parmar, Mihir and Patel, Nisarg and Varshney, Neeraj and Nakamura, Mutsumi and Luo, Man and Mashetty, Santosh and Mitra, Arindam and Baral, Chitta},
  title     = {{LogicBench}: Towards Systematic Evaluation of Logical Reasoning Ability of Large Language Models},
  booktitle = {Proceedings of the 62nd Annual Meeting of the Association for Computational Linguistics (Volume 1: Long Papers)},
  pages     = {13679--13707},
  year      = {2024},
  publisher = {Association for Computational Linguistics},
  doi       = {10.18653/v1/2024.acl-long.739},
  url       = {https://aclanthology.org/2024.acl-long.739/}
}

@inproceedings{zhou2025rulearena,
  author    = {Zhou, Ruiwen and Hua, Wenyue and Pan, Liangming and Cheng, Sitao and Wu, Xiaobao and Yu, En and Wang, William Yang},
  title     = {{RuleArena}: A Benchmark for Rule-Guided Reasoning with {LLM}s in Real-World Scenarios},
  booktitle = {Proceedings of the 63rd Annual Meeting of the Association for Computational Linguistics (Volume 1: Long Papers)},
  pages     = {550--572},
  year      = {2025},
  publisher = {Association for Computational Linguistics},
  doi       = {10.18653/v1/2025.acl-long.27},
  url       = {https://aclanthology.org/2025.acl-long.27/}
}

@inproceedings{scholak2021picard,
  author    = {Scholak, Torsten and Schucher, Nathan and Bahdanau, Dzmitry},
  title     = {{PICARD}: Parsing Incrementally for Constrained Auto-Regressive Decoding from Language Models},
  booktitle = {Proceedings of the 2021 Conference on Empirical Methods in Natural Language Processing},
  pages     = {9895--9901},
  year      = {2021},
  publisher = {Association for Computational Linguistics},
  doi       = {10.18653/v1/2021.emnlp-main.779},
  url       = {https://aclanthology.org/2021.emnlp-main.779/}
}

@article{dong2024xgrammar,
  author  = {Dong, Yixin and Ruan, Charlie F. and Cai, Yaxing and Lai, Ruihang and Xu, Ziyi and Zhao, Yilong and Chen, Tianqi},
  title   = {{XGrammar}: Flexible and Efficient Structured Generation Engine for Large Language Models},
  journal = {Proceedings of Machine Learning and Systems},
  volume  = {7},
  year    = {2025},
  url     = {https://proceedings.mlsys.org/paper_files/paper/2025/hash/5c20ca4b0b20b0bd2f1d839dc605e70f-Abstract-Conference.html}
}

@article{li2023scallop,
  author  = {Li, Ziyang and Huang, Jiani and Naik, Mayur},
  title   = {Scallop: A Language for Neurosymbolic Programming},
  journal = {Proceedings of the ACM on Programming Languages},
  volume  = {7},
  number  = {PLDI},
  pages   = {1463--1487},
  year    = {2023},
  doi     = {10.1145/3591280},
  url     = {https://doi.org/10.1145/3591280}
}

@inproceedings{bourgaux2022provenance,
  author    = {Bourgaux, Camille and Bourhis, Pierre and Peterfreund, Liat and Thomazo, Micha{\"e}l},
  title     = {Revisiting Semiring Provenance for Datalog},
  booktitle = {Proceedings of the 19th International Conference on Principles of Knowledge Representation and Reasoning},
  pages     = {91--101},
  year      = {2022},
  doi       = {10.24963/kr.2022/10},
  url       = {https://doi.org/10.24963/kr.2022/10}
}

@inproceedings{sun2024determlr,
  title={Determ{LR}: Augmenting {LLM}-Based Logical Reasoning from Indeterminacy to Determinacy},
  author={Sun, Hongda and Xu, Weikai and Liu, Wei and Luan, Jian and Wang, Bin and Shang, Shuo and Wen, Ji-Rong and Yan, Rui},
  booktitle={Proceedings of the 62nd Annual Meeting of the Association for Computational Linguistics (Volume 1: Long Papers)},
  pages={9828--9862},
  year={2024},
  doi={10.18653/v1/2024.acl-long.531},
  url={https://aclanthology.org/2024.acl-long.531/}
}

@article{li2024lina,
  title={Leveraging {LLM}s for Hypothetical Deduction in Logical Inference: A Neuro-Symbolic Approach},
  author={Li, Qingchuan and Li, Jiatong and Liu, Tongxuan and Zeng, Yuting and Cheng, Mingyue and Huang, Weizhe and Liu, Qi},
  journal={arXiv preprint arXiv:2410.21779},
  year={2024},
  url={https://arxiv.org/abs/2410.21779}
}

@article{li2025hblr,
  title={From Hypothesis to Premises: {LLM}-Based Backward Logical Reasoning with Selective Symbolic Translation},
  author={Li, Qingchuan and Cheng, Mingyue and Liu, Zirui and Wang, Daoyu and Zeng, Yuting and Liu, Tongxuan},
  journal={arXiv preprint arXiv:2512.03360},
  year={2025},
  url={https://arxiv.org/abs/2512.03360}
}

@article{sviridenko2004,
  author = {Sviridenko, Maxim},
  title = {A Note on Maximizing a Submodular Set Function Subject to a Knapsack Constraint},
  journal = {Operations Research Letters},
  volume = {32},
  number = {1},
  pages = {41--43},
  year = {2004},
  doi = {10.1016/S0167-6377(03)00062-2}
}

@misc{mental,
  title={Are LLMs Stable Formal Logic Translators in Logical Reasoning Across Linguistically Diversified Texts?},
  author={Li, Qingchuan and Li, Jiatong and Liu, Zirui and Cheng, Mingyue and Zeng, Yuting and Liu, Qi and Liu, Tongxuan},
  year={2026},
  eprint={2506.04575},
  archivePrefix={arXiv},
  note={Version 3; accepted by The Web Conference 2026},
  url={https://arxiv.org/abs/2506.04575}
}
\endgroup
\bibliographystyle{iclr2027_conference}

\appendix
\section{Full Related Work}
\label{app:related}
This appendix expands the comparison in Section~\ref{sec:related}, covering neural representations of logic, symbolic execution, premise selection, and source checking.

\subsection{First wave: connectionist--symbolic integration (1990s)}
KBANN \cite{kbann} compiled a propositional Horn rule set into the
topology of a feed-forward network and then refined the rules by
backpropagation. It supported both rule insertion and extraction from a neural model.
CILP \cite{cilp} showed that recurrent networks with semi-linear
neurons approximate the fixed-point operator of propositional logic
programs, closing a learn--reason--extract cycle with sound extraction
\cite{garcez2022}. These systems established ways to encode symbolic rules in neural networks. CPUNeSy builds on the separation between rules and their neural representation, with additional checks on source premises and their storage.

\subsection{Statistical relational learning and probabilistic logics
(2000s)}
Markov logic networks \cite{mln} softened first-order clauses with
weights; ProbLog \cite{problog} gave logic programs a distribution
semantics. These methods represent uncertainty over logical statements. Their use of probabilistic semantics addresses a different aspect of inference from the source-admission checks studied here.

\subsection{Differentiable deductive engines (2016--2022)}
TensorLog \cite{tensorlog} compiled Datalog-style inference into
differentiable matrix algebra; end-to-end differentiable proving
\cite{e2eprover} replaced unification by vector similarity; $\partial$ILP
\cite{dilp} learned rules by gradient descent; DeepProbLog
\cite{deepproblog} introduced neural predicates over probabilistic
facts; NeurASP \cite{neurasp} embedded networks into answer-set
programs; SLASH \cite{slash} combined neural predicates with
probabilistic circuits. These works made predicate-level processing
\emph{trainable}. The learned predicates connect neural predictions to logical inference. Reasoning shortcuts can arise when those predictions fit the task labels while misrepresenting the intended concepts (Appendix~\ref{subsec:shortcuts}).

\subsection{Logic as layers and constraints}
Logical Neural Networks \cite{lnn} realize logical gates as neurons
with provable bound propagation; Logic Tensor Networks \cite{ltn}
ground fuzzy first-order formulas in tensor computations; SATNet
\cite{satnet} made MaxSAT a differentiable layer; the semantic loss
\cite{semloss} moves the constraint into the training objective. These methods incorporate logic into network computation or training objectives. CPUNeSy uses a separate symbolic state whose updates are checked before execution.

\subsection{LLM-era hybrids}
The closest systems translate model output into a symbolic problem. Logic-LM \cite{logiclm}, LINC \cite{olausson2023linc},
SatLM \cite{ye2023satlm}, LeanReasoner \cite{jiang2024leanreasoner},
and CLOVER \cite{ryu2025clover} all make the model translate natural
language into a formal object and then rely on a symbolic executor or
verifier. Their target formalisms include logic programming, FOL, SAT/SMT, and Lean. They also differ in the repair and decomposition performed before execution. SymbCoT \cite{xu2024symbcot} and Aristotle
\cite{xu2025aristotle} keep more of the symbolic reasoning trajectory
inside LLM prompting, while CaRing \cite{yang2025caring} extracts
human-readable proofs from Prolog search logs. These systems motivate the distinction between validating an inference and checking its premises. CPUNeSy admits source units through a certificate gate and records the dependencies of conclusions produced by the symbolic core.

Structured generation addresses errors in output syntax.
PICARD \cite{scholak2021picard} and XGrammar
\cite{dong2024xgrammar} constrain decoding so generated strings remain
syntactically valid. This is useful for reducing parser failures, but a
well-formed formula may still assert an unsupported fact. Datalog
engines such as Souffl\'e \cite{souffle} and provenance semantics for
Datalog \cite{bourgaux2022provenance} provide execution and dependency tracking. CPUNeSy combines these operations with source admission and removal of redundant stored units. Proposition~\ref{prop:scheduling} specifies when to cache a derived result under the stated cost model, and Section~\ref{subsec:rq4} reports the measured costs.

\subsection{Neural theorem proving and premise selection}
GPT-f \cite{gptf} introduced transformer step-proving on Metamath;
HTPS \cite{htps} added search-guided expert iteration; DSP \cite{dsp}
pipelines informal drafts to formal sketches; LeanDojo/ReProver
\cite{leandojo} and MagnusHammer \cite{magnushammer} made learned
premise selection the standard interface to large libraries;
AlphaGeometry \cite{alphageometry} demonstrated a neuro-symbolic loop
at olympiad level; \cite{dl4tp} surveys the area. Premise-selection indexes and our kernel store both supply candidates for proof construction. We compare candidate coverage and selection accuracy against ReProver in Appendix~\ref{subsec:leandojo}.

\subsection{Failure modes: reasoning shortcuts}
\label{subsec:shortcuts}
Marconato et al.\ \cite{shortcuts} showed that NeSy predictors can
reach high label accuracy while violating the intended semantics;
mitigation requires extra supervision or architectural constraints.
These findings motivate checking which premises support a prediction. Our experiments measure the effects of restricted writes, agreement, and source rechecking (Section~\ref{subsec:rq2}).

\subsection{Logical-reasoning benchmarks and pressure tests}
Recent benchmarks sharpen the empirical boundary for LLM reasoning.
LogicBench \cite{parmar2024logicbench} isolates individual inference
rules across propositional, first-order, and non-monotonic patterns;
ProverGen/ProverQA \cite{qi2025provergen} combines LLM generation with
symbolic provers to create FOL reasoning problems with verified
intermediate chains; RuleArena \cite{zhou2025rulearena} moves rule
following into realistic airline, sports, and tax-policy settings.
These benchmarks test whether models follow contextual rules through multiple inference steps. Our L1-HARD, M-HARD, ContractNLI, MedCalc, and LeanDojo evaluations examine how deterministic execution and checks on model-proposed premises affect errors, abstentions, and cost.

\subsection{Positioning}
Table~\ref{tab:positioning} identifies the role of related techniques in
the present design and pairs each contribution with the evaluation question studied here.

\begin{table}[!htbp]
\centering\footnotesize
\caption{Relationship to established components and evaluation questions.}
\label{tab:positioning}
\begin{tabularx}{\linewidth}{@{}lXX@{}}
\toprule
Family & Established contribution & Question tested here\\
\midrule
Parse-and-solve & Execute model-generated formalizations & Which agreed outputs does source rechecking withhold?\\
Self-consistency & Aggregate repeated predictions & What changes when the two traces and votes are held fixed?\\
Structured decoding & Enforce syntactic output contracts & Can an allowed semantic assertion still be wrong?\\
Proof/provenance systems & Check inference and trace dependencies & Which premises may enter the checked state?\\
Selective prediction & Trade coverage against answered risk & When do error, refusal, and call costs favor checking?\\
Kernel and cache methods & Reduce stored or recomputed state & Does this corpus offer sufficient redundancy to justify compression?\\
\bottomrule
\end{tabularx}
\end{table}
CPUNeSy uses deterministic Horn execution and proof/provenance ideas
\citep{souffle,gptf,dsp,leandojo}, and retains the stated cost assumptions
of the depth--information model \citep{derdepth}. Its primary empirical
addition is the explicit decomposition of restricted writes, agreement,
and source rechecking with observable error and refusal counts.
\section{Background, Definitions, and Supporting Results}
\label{app:background}

This appendix fixes the notation and states the background results
invoked by CPUNeSy. Section~\ref{sec:theory} states the main guarantees, and Appendix~\ref{app:proofs} gives their proofs.

\subsection{States, Formulas, and Deductive Closure}
Following the Objective Information Theory (OIT) foundation \cite{oit},
the \emph{state} of an object or system at a time point (or interval) is
the semantic valuation of a set of well-formed formulas over a given
domain and interpretation. The framework relates a \emph{semantic domain} $S_O$ of statements to a \emph{carrier domain} $S_C$ that represents them \cite{derdepth}. In CPUNeSy, the language model is the carrier and the symbolic core operates on the statements it proposes.

\begin{definition}[Statement universe and proof system \cite{semrd}]
\label{def:universe}
Fix a finite ambient universe $\SO$ of candidate statements and an
effective proof system $\PS$ such that for every finite
$\Gamma\subseteq\SO$ and every $s\in\SO$, the judgment
$\Gamma\vdash s$ is decidable. The induced \emph{deductive closure} is
\[
\Cn(\Gamma)\;:=\;\{\,s\in\SO : \Gamma\vdash s\,\},
\]
which is reflexive ($\Gamma\subseteq\Cn(\Gamma)$), monotone, and
idempotent ($\Cn(\Cn(\Gamma))=\Cn(\Gamma)$), with $\Cn(\Gamma)$ finite
for finite $\Gamma$.
\end{definition}

\begin{assumption}[Finite-step closure dynamics \cite{semrd}]
\label{ass:finite-step}
There is a computable monotone operator $\TPS$ on finite subsets of
$\SO$ with $\TPS(B)\supseteq B$ and
$\Cn(B)=\bigcup_{n\ge 0}\TPS^{n}(B)$, stabilizing in finitely many
iterations. Function-free Horn / Datalog fragments satisfy
Assumption~\ref{ass:finite-step}; this is the decidable fragment in which
CPUNeSy cores operate (criterion C1 below).
\end{assumption}

\subsection{Deductive Sources and the Irredundant Kernel}
\begin{definition}[Deductive source \cite{semrd}]
\label{def:dedsource}
A \emph{deductive source} is a pair $(S_O,P_O)$ where
$S_O\subseteq\SO$ is a finite, effectively listable knowledge base under
a fixed canonical order, and $P_O$ is a distribution over $S_O$. Fidelity is defined by preservation of the statements that can be derived.
\end{definition}

\begin{definition}[Kernel decomposition \cite{semrd}]
\label{def:kernel}

Fixing $\PS$ and a canonical scan order, a deterministic deletion scan
decomposes
\[
S_O \;=\; A \uplus J,
\]
where $A=\Atom(S_O)$ is the \emph{irredundant kernel} and $J$ the
\emph{redundant part}: every $j\in J$ satisfies
$j\in\Cn(S_O\setminus\{j\})$. Under the order-robustness condition
$A=\Ess(S_O):=\{s\in S_O: s\notin\Cn(S_O\setminus\{s\})\}$, the kernel is
order-invariant. We write $P_A:=P_O(A)$ for the \emph{kernel mass} and
$\pi_A$ for the source conditioned on $A$; the \emph{redundancy
fraction} is $1-P_A$. We use \emph{weighted} kernel mass $P_A=P_O(A)$ for the service distribution and \emph{uniform} kernel mass $\PAu:=|A|/|S_O|$ for the unit count.
\end{definition}

\begin{theorem}[Kernel factorization and zero-distortion rate \cite{semrd}]
\label{thm:bg-semrd}
Let $\hat S_O\subseteq\Cn(S_O)\cap\SO$ (reconstructions stay inside the
closure). Then for every $D\ge 0$,
\[
R_{\Sem}(D)\;=\;P_A\cdot R^{(A)}\!\Big(\tfrac{D}{P_A}\Big),
\]
i.e.\ the redundant part is invisible to both rate and distortion. In the
nonconfusable (disjoint-kernel) regime,
$R_{\Sem}(0)=P_A\,H(\pi_A)$; in the general confusable case,
$R_{\Sem}(0)=P_A\,H_{\Gamma_0}(\pi_A)$ where $H_{\Gamma_0}$ is the
entropy of the kernel confusability hypergraph. The corresponding
nonconfusable per-message specialization used by CPUNeSy is stated
in Theorem~\ref{thm:TA}.
\end{theorem}

Theorem~\ref{thm:bg-semrd} relates the storage requirement to the kernel mass under closure fidelity. 


\subsection{Derivation Depth as an Information Metric}
\begin{definition}[Derivation depth \cite{derdepth}]
\label{def:depth}

For a finite, effectively decidable premise base $B\subseteq\SO$ with
$m:=|B|$, let $P_O(s)$ denote the (finite, computable, well-founded)
immediate-predecessor set of $s$. The \emph{derivation depth}
$\Dd(s\mid B)$ is the length of a shortest derivation of $s$ from $B$
along $P_O$; it is a unique, finite, computable non-negative integer.
\end{definition}

\begin{theorem}[Depth as information \cite{derdepth}]
\label{thm:bg-depthinfo}
Under a richness condition and generic incompressibility, for
information-rich queries with $d:=\Dd(q\mid B)\ge 1$,
\[
\Dd(q\mid B)\;=\;\widetilde{\Theta}\!\left(
\frac{K\!\big(q\,\big|\,\langle B\rangle\big)}
{\log\big(m+\Dd(q\mid B)\big)}\right),
\]
where $K(\cdot\mid\langle B\rangle)$ is conditional algorithmic
description length and $\widetilde{\Theta}$ hides an unavoidable
logarithmic addressing overhead. The characterization is tight, and
derivation depth coincides with conditional Bennett logical depth under
efficient proof simulation.
\end{theorem}

\paragraph{From depth to caching.}
The store--compute model of \citet{derdepth} supplies the cost scaling
used in Proposition~\ref{prop:scheduling}. The proof of its caching threshold and allocation guarantee appears in
Appendix~\ref{app:proof-scheduling}. The reconstruction-depth bound
used after pruning is stated in Lemma~\ref{lem:cert-depth}.

Theorem~\ref{thm:bg-depthinfo} provides the background
depth--information characterization; Proposition~\ref{prop:scheduling}
uses the stationary cost model to decide when recomputation gives way
to caching.

\subsection{Proof-Valid Reliability Under Premise Erasure}
\begin{definition}[Premise erasures and transparent caches \cite{pvcache}]
\label{def:erasure}
Base leaves are independently unavailable with rate $\varepsilon$ due to storage loss. Their source support remains valid. A fixed
cache $X\subseteq\Cn(B)$ remains available during the erasure trial.
Recovery is evaluated on a fixed hereditary witness DAG, with each inference requiring all its designated parents. Only this witness is used for recovery. A leaf is residually
exposed if some path from it to $q$ avoids $X$. Write $\Dexp(q;X)$
for these leaves and $\Dexp(q):=\Dexp(q;\varnothing)$.
\end{definition}

The residual-leaf law is stated in
Corollary~\ref{cor:erasure}, and proved in
Appendix~\ref{app:proof-erasure}. The following background result
records the additional module and complexity claims.

\begin{theorem}[Module gains and cache-selection complexity \cite{pvcache,pvbench}]
\label{thm:bg-modules}
For shared workloads, \emph{semantic modules} yield exact reliability laws
under joint and maximal-error criteria; leaf-only transparent storage
incurs a first-order overhead factor $1/\varepsilon$, reduced in the
saturated shared regime to $\rho/(s\varepsilon)$, where $s$ premises are
protected by one module and $\rho$ is the module-to-leaf cost ratio.
Optimal selection in general derivation DAGs is NP-complete already at
depth two; MDS parity caching is optimal up to one packet against the
coded benchmark.
\end{theorem}


\subsection{Verification of Model-Proposed Derivations}
\label{subsec:bridge}

The preceding background concerns deductive sources. The definitions
below specify how a carrier trace is mapped to a certified source.
Voting is stated in Section~\ref{sec:theory}; coverage, certified depth,
and scheduling appear in Appendix~\ref{app:resource-guarantees}.
This subsection supplies their common setup.

\begin{definition}[Bridge channel model]
\label{def:bridge-channel}
The carrier--core interface is modeled by three error channels:
(i) \emph{proposal unreliability} $\mu$: each step of the intended
derivation is corrupted or missing in the carrier trace with
probability at most $\mu$;
(ii) \emph{verification miss rate} $\eta$: a valid step fails
verification with probability at most $\eta$;
(iii) \emph{formalization-failure rate} $\epsfr$: each exposed premise
fails formalization with probability at most $\epsfr$. Channel (iii) consists of an interface-dependent bias component and an i.i.d.\ slip component
(Definition~\ref{def:bias-slip}), and the independence of channel
events (Assumption~\ref{ass:standing}, S3) applies to the slip
component only.
\end{definition}

\begin{definition}[Grounding interface]
\label{def:grounding-interface}
A \emph{grounding interface specification} is a triple
$\mathcal{I}=(\Sigma,\mathcal{L},\Omega)$, where $\Sigma$ is a scoped
predicate vocabulary with explicit signatures (arity and argument
roles), $\mathcal{L}$ is a role legend, i.e.\ a partial injective map
from constants to case/entity roles, and $\Omega\subseteq\Sigma$ is a
declared source-predicate whitelist whose ground instances the
carrier may propose directly. The whitelist determines which predicates the model may propose. Their instances still require source checking. A \emph{grounding operator}
$\tau_{\mathcal{I}}$, implemented by the carrier, maps each query $x$
to a set $\hat G(x)$ of ground atoms over $\Sigma$; $\tau_{\mathcal{I}}$
implements the vocabulary and role constraints specified by $\mathcal{I}$.
\end{definition}

\begin{definition}[Bias--slip decomposition]
\label{def:bias-slip}

Fix an interface $\mathcal{I}$, a query distribution, and a gold
grounding $G(x)$ per query. The $i$-th carrier call produces
$\hat G_i(x)=G(x)\oplus B_{\mathcal{I}}(x)\oplus S_i(x)$ (symmetric
difference of ground-atom sets), where the \emph{bias} component
$B_{\mathcal{I}}(x)$ is determined by $\mathcal{I}$ and hence identical
across calls, and the \emph{slip} components $S_i(x)$ are i.i.d.\
across calls. A component is \emph{answer-critical} for $x$ if it
changes the answer derived from the grounded premises. Write
$\bans(\mathcal{I}):=\Pr_x[B_{\mathcal{I}}(x)\text{ is answer-critical}]$
and $\slipc:=\Pr_x[S_i(x)\text{ is answer-critical}]$. The
formalization-failure rate of Definition~\ref{def:bridge-channel}(iii)
satisfies $\epsfr(\mathcal{I})\le\bans(\mathcal{I})+\slipc$, with
equality on disjoint error events; the bias component is perfectly
correlated across traces, the slip component independent (S3).
\end{definition}

\begin{definition}[Certified deductive rounding]
\label{def:rounding}

Given premise base $B$, a trace $\tr$, an edge extractor $E$, and a
verifier $V$ for a decidable fragment, the verification procedure, called \emph{certified rounding} and denoted $\mathcal{R}_{E,V}(B,\tr)$, deletes every unverifiable edge, discards conclusions without a complete
verified derivation from $B$, and returns the certified source $(\Scert,\Pcert)$. A \emph{certificate gate} admits
into the core only units whose derivation edges all verify.
\end{definition}

\begin{definition}[Certified depth; recovery depth]
\label{def:cert-depth}
$\Ddc(q\mid B)$ is the minimum derivation depth of $q$ over derivations
whose every edge verifies under $V$. Letting $\pi$ be a minimum-depth
derivation of $q$, the \emph{recovery depth} $\drec(q)$ is the maximum,
over edges of $\pi$ pruned by rounding, of the certified depth needed to
re-derive the pruned subgoal from $B$.
\end{definition}

\begin{definition}[Canonical cache key]
\label{def:ckey}
Under the canonical order of Definition~\ref{def:kernel}, every
predicate unit $u$ has a unique normal form; $\key(u)$ is its
serialization. Units colliding under $\key$ are logically equivalent,
so the key is a sound cache index for certified content.
\end{definition}

\begin{lemma}[Certified rounding yields a deductive source]
\label{lem:rounding-source}
For a finite carrier trace and a sound, decidable fragment verifier,
whenever the rounding operator of Definition~\ref{def:rounding}
retains a nonempty set of units, its output
$(\Scert,\Pcert)$ is a deductive source in the sense of
Definition~\ref{def:dedsource}.
\end{lemma}
The proof is given in Appendix~\ref{app:proof-rounding}.
If no unit survives verification, the gate rejects the proposal.

\paragraph{Connection to the main guarantees.}
Proposition~\ref{prop:coverage} gives the coverage guarantee,
Lemma~\ref{lem:cert-depth} bounds certification overhead,
Lemma~\ref{lem:voting-floor} separates shared bias from independent
slips, and Proposition~\ref{prop:scheduling} governs caching.
Their proofs appear in
Appendices~\ref{app:proof-coverage}, \ref{app:proof-depth},
\ref{app:proof-voting}, and~\ref{app:proof-scheduling}, respectively.

\begin{remark}[Observed joint-error diagnostic]
\label{rem:voting-floor-measurement}
The joint outcomes describe how errors co-occur across calls, the behavior considered in Lemma~\ref{lem:voting-floor}. Broad-interface L1-HARD has $9$ both-wrong pairs; the restricted interface has none and four correct/wrong pairs. Two calls cannot identify $\bans$ and $\slipc$ separately or establish independence. The restricted vocabulary blocks direct conclusion proposals, while false allowed leaf assertions remain possible (Appendix~\ref{app:gate-audit}).
\end{remark}

The interface and rounding operator determine what can enter the
certified representation. Section~\ref{sec:theory} states the guarantees
for source admission and certified derivation, together with the voting bound.
Appendix~\ref{app:resource-guarantees} gives the coverage, kernel-storage,
caching, and recovery guarantees.

\subsection{Applicability Criteria}

\begin{definition}[Four-criterion compression applicability]
\label{def:four-criteria}
A knowledge base $\mathcal{K}$ is \emph{eligible for the proposed compression-and-caching regime} iff:
(C1) $\mathcal{K}$ is formalizable in a decidable fragment (e.g.\
function-free Horn/Datalog);
(C2) task fidelity is defined by the consequences entailed by the stored statements;
(C3) the measured redundancy fraction satisfies $1{-}P_A\ge\tau$ for a
protocol-fixed threshold $\tau$;
(C4) the derivation-depth distribution admits a budget compatible with
finite-step closure, with query frequencies spanning the $\fc$ threshold so
the admission law is operative.
The supporting results are Lemma~\ref{lem:rounding-source} for the
certified source, Theorem~\ref{thm:TA} for kernel storage, and
Proposition~\ref{prop:scheduling} for scheduling. The criteria are evaluated by the protocols in Appendix~\ref{app:protocols}. C3 specifies the target compression yield. Checked execution and closure preservation also apply below this threshold.
\end{definition}

\section{Coverage, Kernel, and Resource Guarantees}
\label{app:resource-guarantees}
The following results describe query coverage, kernel storage, caching, and recovery after storage loss. Each result states its assumptions. Appendix~\ref{app:companion-proofs} identifies the rate--distortion and scheduling models used in the proofs.
\begin{proposition}[Certified coverage]
\label{prop:coverage}
Condition on the absence of interface-wide grounding bias. Suppose proposal loss, rejection of a valid edge, and premise-formalization failure occur independently in the slip channel with probabilities at most $\mu$, $\eta$, and $\epsfr$. If a query requires $n$ derivation edges and $e$ exposed source premises, then
\[
\Pr[\text{$q$ has a complete certified derivation}]
\;\ge\;(1-\mu)^n(1-\eta)^n(1-\epsfr)^e .
\]
\end{proposition}
The proof appears in Appendix~\ref{app:proof-coverage}. This bound formalizes the derivation-demand condition: local failures compound with path length, while correlated grounding bias is handled separately below.

\begin{lemma}[Certified-depth overhead]
\label{lem:cert-depth}
For any query admitting a certified reconstruction,
\[
\Dd(q\mid B)\le\Ddc(q\mid B)\le\Dd(q\mid B)+\drec(q),
\]
where $\drec(q)$ is the largest additional depth needed to reconstruct a sub-derivation removed by certification.
\end{lemma}
The proof appears in Appendix~\ref{app:proof-depth}.
The bound gives the additional derivation depth required after verification.

\begin{theorem}[Kernel sufficiency under closure fidelity]
\label{thm:TA}
Assume a decidable closure operator and a fixed canonical deletion order. Then (i) $\Cn(A)=\Cn(\SO)$; (ii) every closure-type query answerable from $\SO$ is answerable from $A$, with certified depth increasing by at most the re-derivation depth of the removed units; and (iii) storing any $J'\subseteq J$ does not improve the achievable rate--distortion function under closure fidelity. Hence the persistent source store needs $|A|=\PAu|\SO|$ units, a fraction $\PAu$ of the full source. For nonempty $A$, uniform messages supported on $A$ in the nonconfusable regime have an asymptotic zero-distortion rate of $\log|A|$ nats per message; uniform symbolwise-exact messages over $\SO$ require $\log|\SO|$. The two rates apply to the respective message distributions just specified.
\end{theorem}
Each deletion removes only a unit derivable from the remainder,
which explains why kernelization preserves closure. The proof
appears in Appendix~\ref{app:proof-kernel}. The guarantee applies
only to queries covered by the certified source. The fraction
$1-\PAu$ measures how much symbolic mass can be removed before
serving, which the corpus audit reports as criterion C3.

\begin{proposition}[Store--recompute threshold]
\label{prop:scheduling}
Under the stationary depth--information cost model of \citet{derdepth}, with store-to-recompute rate ratio $\rho$ and effective kernel size $m_{\mathrm{eff}}$, the break-even frequency for caching a certified unit satisfies
\[
\fc(u)=\widetilde{\Theta}\!\left(\rho\log\!\left(m_{\mathrm{eff}}+\Ddc(u)\right)\right).
\]
Under a storage budget and a normalized monotone depth-reduction objective with diminishing returns, partial-enumeration greedy attains a $1-1/e$ approximation \citep{sviridenko2004}.
\end{proposition}
The proof appears in Appendix~\ref{app:proof-scheduling}. The threshold selects cache candidates, and the allocation rule selects a subset within the storage budget. The logarithmic scale is inherited from the cited depth--information model, while the allocation guarantee additionally requires the stated submodularity condition.
Kernelization removes units that do not change closure; caching stores selected consequences that do change runtime cost, deep or frequently used ones first.

\begin{corollary}[Reliability under premise erasure]
\label{cor:erasure}
Under the fixed-cache availability model of Definition~\ref{def:erasure}, independent premise erasure at rate $0<\varepsilon<1$, and a canonical-witness derivation DAG, a query with $\ell$ residually exposed leaves remains recoverable with probability $(1-\varepsilon)^\ell$. Caching a valid chain segment improves this probability exactly when it removes at least one leaf from the residual exposed set.
\end{corollary}
The proof appears in Appendix~\ref{app:proof-erasure}.

The three accuracy conditions concern derivation demand (Proposition~\ref{prop:coverage}), certified-source coverage (Theorem~\ref{thm:TA}), and slip- versus bias-dominated grounding (Lemma~\ref{lem:voting-floor}). Separately, the corpus audit uses criteria C1--C4. On the $14{,}414$-unit AFP corpus and $503$-unit regulatory corpus, C1, C2, and C4 pass. For C3, bounded $k$-step rederivation saturates at $37.1\%$ and $0.60\%$ for $k\le4$, below $\tau=0.5$ (Table~\ref{tab:kstep}). These measurements favor retaining a larger source store. In the regulatory corpus, redundancy is concentrated in the rule layer, and the composition law matches the observed corpus-wide value within half a point.


\section{Proofs and Auxiliary Results}
\label{app:proofs}

This appendix proves the results in Section~\ref{sec:theory} and Appendix~\ref{app:resource-guarantees}. Auxiliary results are stated before their proofs. Appendix~\ref{app:companion-proofs} lists the background results used in these arguments.

\subsection{Standing Assumptions}

The following conditions collect the notation used by the auxiliary
arguments. Each result specifies which conditions it requires.

\begin{assumption}
\label{ass:standing}

\begin{enumerate}[label=S\arabic*.]
  \item \textbf{Decidable fragment with a sound and complete verifier.}
        The core's fragment satisfies Assumption~\ref{ass:finite-step},
        and $V(e,\Gamma)=\top$ iff $e$ is a valid inference step over
        $\Gamma$ within the fragment.
  \item \textbf{Fixed canonical order.} The canonical scan order of
        Definition~\ref{def:kernel} is public and fixed; kernel
        extraction is deterministic.
  \item \textbf{Channel independence.} The bridge channel errors
        (proposal failure rate $\mu$, verification miss rate $\eta$,
        formalization failure rate $\epsfr$) are independent across
        derivation steps and exposed premises. For the
        formalization channel this applies to the slip component only;
        the interface-determined bias component is perfectly correlated
        across traces and is modeled separately
        (Definition~\ref{def:bias-slip}, Lemma~\ref{lem:voting-floor}).
  \item \textbf{Stationary amortized cost model.} Stationary request
        stream; storage rate $\alpha$; recomputation rate $\beta$;
        $\rho=\alpha/\beta$; storage cost $\sigma$ normalized to one
        step-equivalent per unit.
  \item \textbf{Canonical-witness erasure regime.} Designated-parent map,
        hereditary derivation DAG (Definition~\ref{def:erasure});
        premise erasures i.i.d.\ with rate $\varepsilon$.
\end{enumerate}
\end{assumption}

\subsection{Certified-State Soundness}
\label{app:proof-state}

\begin{proof}[Proof of Theorem~\ref{thm:state-soundness}]
Write $\Gamma_t=A_t\cup\mathcal R_t$. We induct on completed state
transitions; served consequences are evaluated against the state at
their serving time. Initialization supplies source-support records for
every kernel unit and an empty cache, so all three claims hold.

\emph{Source admission.} Rejected proposals change no certified state.
An accepted unit is type-checked and has a logged source-support
record. The deletion scan retains a subset of these source-checked
units, hence each new kernel premise still has its record. A newly supported premise may add information beyond the previous kernel. Kernel changes trigger cache
revalidation. Each retained entry has a complete verified witness
over $\Gamma_{t+1}$ of depth at most $\delta$; all other entries are
evicted. Soundness of the fragment verifier proves its conclusion is
in $\Cn(\Gamma_{t+1})$, and the minimum certified depth is no greater
than the depth of this witness. Thus (i) and (ii) are preserved.

\emph{Derived admission and queries.} The gate requires a complete
acyclic witness rooted at the proposed conclusion, with source leaves
in $\Gamma_t$. Induction over its topologically ordered, verified
inference steps gives $u\in\Cn(\Gamma_t)$. Its checked witness depth
$d\le\delta$ implies $\Ddc(u\mid\Gamma_t)\le d\le\delta$.
Serving or caching this unit changes no source premise. A cache hit
reuses a certificate validated for the current kernel version, so the
same argument applies. A rejected or uncertified fallback answer is
outside the certified serving channel.

\emph{Erasure and eviction.} Erasing a source identifier removes it
from the active certified source set and cascade-evicts every entry
whose transitive certificate dependencies include that identifier.
This proves (iii), including dependencies through cached intermediate
results. Recomputing the kernel and revalidating any surviving cache
entries preserves (i) and (ii) by the source-admission argument.
Cache eviction removes an entry and any entries depending on it;
source eviction follows the erasure rule. Neither operation admits
unsupported content. Frequency updates change no certified content.
These cases exhaust the transitions and complete the induction.
\end{proof}

\paragraph{Applying the coverage and voting results.}
For a query whose reconstruction fits the serving budget,
Proposition~\ref{prop:coverage} and Lemma~\ref{lem:cert-depth}
supply the coverage and depth guarantees. Voting is governed
separately by Lemma~\ref{lem:voting-floor}. 

\subsection{Certified Rounding as a Deductive Source}
\label{app:proof-rounding}

\begin{proof}[Proof of Lemma~\ref{lem:rounding-source}]
$\Scert$ collects the conclusions of extracted derivation DAGs all
of whose edges verify. It is finite, since the carrier trace is
finite, and effectively listable, since edge verification is
decidable. Its nonemptiness permits the empirical distribution
$\Pcert$ to be normalized. Soundness of $V$ gives
$\Scert\subseteq\Cn(B)$. Thus the finite, effectively listable
source with its empirical distribution satisfies
Definition~\ref{def:dedsource}.
\end{proof}

\subsection{Certified Coverage}
\label{app:proof-coverage}

\begin{proof}[Proof of Proposition~\ref{prop:coverage}]
Work on the bias-free query mass specified in the proposition.
For the intended derivation to survive rounding, it suffices that
each of its $n$ steps be proposed, each of those steps pass
verification, and each of its $e$ exposed premises be formalized.
The respective probabilities are at least $(1-\mu)^n$,
$(1-\eta)^n$, and $(1-\epsfr)^e$. Independence in the slip channel
makes the probability of their intersection at least their product.
On this event, rounding retains a complete certified derivation,
which proves the claimed lower bound.
\end{proof}

\subsection{Certified-Depth Overhead}
\label{app:proof-depth}

\begin{proof}[Proof of Lemma~\ref{lem:cert-depth}]
Certified derivations form a subset of all derivations, giving
the lower bound. For the upper bound, let $\pi$ attain
$\Dd(q\mid B)$. Restore each subgoal removed by rounding using
its shortest certified reconstruction from $B$. Under the
recovery-depth accounting of Definition~\ref{def:cert-depth},
this adds at most $\drec(q)$ to the derivation depth. The resulting
derivation is certified, giving the upper bound.
\end{proof}

\subsection{Voting Under Shared Bias and Independent Slips}
\label{app:proof-voting}

\begin{proof}[Proof of Lemma~\ref{lem:voting-floor}]
Condition on a query in the shared-error regime, whose total mass is $\bans$. By the lemma's explicit assumption, indicators that a trace certifies the common wrong answer are conditionally independent with means at least $1-\slipc>1/2$.
Hoeffding's inequality \cite{hoeffding} bounds the probability of
failing to obtain a strict majority of that answer by
$e^{-2k(1/2-\slipc)^2}$. Multiplication by $\bans$ gives the lower
bound in the lemma.

For a query outside that regime, let $X_i$ indicate that trace $i$
certifies the unique correct answer. By the lemma's assumptions,
the $X_i$ are independent with means at least $1/2+\delta_0$.
Failure to obtain a strict correct majority is contained in
$\{\sum_i X_i\le k/2\}$, whose probability is at most
$e^{-2k\delta_0^2}$ by Hoeffding's inequality. Adding the
worst-case contribution from the biased query mass gives the
stated upper bound after averaging over queries. A missing strict majority yields abstention, including ties. The lower concentration bound requires the stated conditional independence within the shared-error regime.
\end{proof}

\subsection{Kernel Sufficiency and Storage}
\label{app:proof-kernel}

\begin{proof}[Proof of Theorem~\ref{thm:TA}]
\emph{Part (i): closure preservation.}
The deletion scan processes $S_O$ in the canonical order and removes $j$
only when $j\in\Cn(\Gamma\setminus\{j\})$ for the current set $\Gamma$.
We claim $\Cn(\Gamma)$ is invariant under each removal. Inclusion
$\Cn(\Gamma\setminus\{j\})\subseteq\Cn(\Gamma)$ is monotonicity.
Conversely, $j\in\Cn(\Gamma\setminus\{j\})$ gives
$\Gamma\subseteq\Cn(\Gamma\setminus\{j\})$, whence
$\Cn(\Gamma)\subseteq\Cn(\Cn(\Gamma\setminus\{j\}))
=\Cn(\Gamma\setminus\{j\})$ by monotonicity and idempotence of $\Cn$
(Definition~\ref{def:universe}). Induction over the scan terminates at
$A$, so $\Cn(A)=\Cn(S_O)$.

\emph{Part (ii): query preservation.}
Closure equality gives the same closure-type answers from $A$
and $S_O$. A derivation using a removed unit can be reconstructed
by replacing that unit with its derivation from the retained
kernel. The depth increase is accounted for by the re-derivation
depth of the removed units involved, as in
Lemma~\ref{lem:cert-depth}.

\emph{Part (iii): rate--distortion invariance.}
The closure-fidelity factorization in
Theorem~\ref{thm:bg-semrd}, imported from \citet[Theorem~5.1]{semrd}, makes the redundant part
invisible to rate and distortion. Adding $J'\subseteq J$ to the
stored kernel therefore leaves the achievable rate--distortion
function unchanged.

\emph{Storage count and the per-message specialization.}
The unit-count identity follows from
$\PAu=|A|/|S_O|$ in Definition~\ref{def:kernel}. It concerns the
physical number of persistent units. For the final specialization, condition messages on $A$ and take
the uniform law there. Then $P_A=1$ and $H(\pi_A)=\log|A|$;
Theorem~\ref{thm:bg-semrd} gives the asserted asymptotic rate.
Uniform symbolwise-exact coding on $\SO$ has entropy $\log|\SO|$.
This calculation compares different source laws. For the original
source law, the nonconfusable rate remains $P_AH(\pi_A)$, as in \citet[Theorem~3.1]{semrd}.
\end{proof}

\begin{remark}[Scope of the kernel guarantee]
The unit-count consequence of Theorem~\ref{thm:TA} and its final
per-message claim concern different resources. The former counts
persistent units, while the latter uses a uniform distribution on the nonconfusable core alphabet.
\end{remark}

\begin{corollary}[Redundancy in mixed corpora]
\label{cor:composition}
Let a mixed corpus decompose as $U=R\uplus F$, where $R$ is a rule
layer (safe Horn clauses with non-empty bodies) and $F$ an instance
layer of ground units each carrying at least one constant that occurs
in no other unit of $U$ (case-individual constants), with derivability
that of the certified fragment ($\theta$-subsumption between rule
units; forward chaining on ground units). Then every $f\in F$ is
essential, the redundancy status of every rule unit is decided entirely
within $R$, and the uniform-reading redundancy fraction obeys the
composition law
\[
1-\PAu(U)=\frac{|R|}{|U|}\,\bigl(1-\PAu(R)\bigr).
\]
\end{corollary}
\begin{proof}[Proof of Corollary~\ref{cor:composition}]
$\theta$-subsumption compares clauses pairwise, so no ground unit can
declare a rule unit redundant; rule redundancy is therefore computed
within $R$ alone. A ground unit $f$ is derivable only via a rule
instance whose body units are all present; by safety, some body unit of
any such instance contains $f$'s head constants, hence contains the
case-unique constant. Only $f$ mentions that constant. Since every rule has a non-empty body, induction on derivation length shows that no ground unit mentioning it is derivable without $f$. Hence
$f\in\Ess(U)$, $|\Ess(U)|=|\Ess(R)|+|F|$, and the displayed law follows
by counting.
\end{proof}
\begin{remark}[Applicability of the case-uniqueness premise]
The case-individual-constant premise holds for judgment-style corpora,
where each instance introduces fresh party constants. On multi-case
collections with shared parties or corporations it must be
re-established per constant class (e.g., by namespacing case-local
identifiers); otherwise ground facts about a shared entity can become
derivable and leave the essential core.
\end{remark}

\begin{remark}[Order sensitivity of the kernel]
\label{rem:order-sensitivity}
Theorem~\ref{thm:TA} uses a fixed canonical deletion order;
its closure-preservation argument does not require order robustness.
Equality with the order-robust essential set is a separate property
of Definition~\ref{def:kernel}. Appendix~\ref{app:protocols}
reports a gap of $0$ units on the regulatory corpus and
$|A\triangle\Ess(S_O)|=2{,}473$ ($17\%$) on the Isabelle corpus.
The theorem's final per-message claim retains its explicit
nonconfusability qualification.
\end{remark}

\begin{remark}[C3 applies to the rule layer]
\label{rem:c3-rule-layer}
Corollary~\ref{cor:composition} implies that applying C3 to raw mixed
corpora scales rule redundancy by $|R|/|U|$. Under the case-uniqueness assumption, the instance facts are essential. The rule-layer C3 measurement therefore applies to the units extracted during \textnormal{\textsc{Deploy}}
(Algorithm~\ref{alg:deploy}, stages D2--D3). The mixed-corpus audit of
Appendix~\ref{subsec:legal} confirms the law. The measured
overall redundancy $2.2\%$ matches the composition prediction
(rule-layer redundancy $13.8\%$ $\times$ rule-layer share $14.7\%$
$\approx 2.0\%$) within half a percentage point, and the rule-layer
figure is corroborated by a certified re-measurement ($12.87\%$
$[9.1,16.6]$, gated formalization plus $\theta$-subsumption matching).
\end{remark}

\subsection{Store--Recompute Threshold and Allocation}
\label{app:proof-scheduling}

\begin{proof}[Proof of Proposition~\ref{prop:scheduling}]
Use \citet[Theorem~4.1]{derdepth} with baseline equal to the
current certified kernel, $m=m_{\mathrm{eff}}$, and operational depth
$d=\Ddc(u)$. The certified cost model uses that theorem's information-rich, generic regime and depth-range assumptions in addition to stationarity. Comparing amortized
storage cost with recomputation gives the stated threshold, with
certification write costs included in the storage rate $\alpha$.

For allocation, write the expected depth reduction as
$F(X)=\sum_q\nu(q)[n_q(\varnothing)-n_q(X)]$.
Then $F(\varnothing)=0$, and adding cached units cannot increase
the required depth. The assumed diminishing-returns property
makes each summand submodular; positive linear combinations
preserve submodularity. Applying the cited partial-enumeration
greedy guarantee to the resulting budgeted objective gives the
claimed approximation ratio for additive storage costs \cite{sviridenko2004}.
\end{proof}

\paragraph{Application to the persistent core.}
Use the current base size for $m_{\mathrm{eff}}$ and the certified
depth of the candidate unit. In the bounded-rounding regime
$\drec(u)=O(m_{\mathrm{eff}}+\Dd(u\mid B))$,
Lemma~\ref{lem:cert-depth} controls reconstruction overhead.
Proposition~\ref{prop:scheduling} then supplies the candidate threshold and the budgeted allocation rule.

\subsection{Reliability Under Premise Erasure}
\label{app:proof-erasure}

\begin{proof}[Proof of Corollary~\ref{cor:erasure}]
Fix the cache before erasure. In the witness DAG, recovery fails
iff a missing leaf has a cache-free path to the root: tracing a failed
uncached node backwards finds a failed parent and eventually such a
leaf; conversely, a missing leaf propagates failure along any uncached
path. This is the local failure-path argument of
\citet[Lemma~3.1 and Theorem~3.1]{pvcache}. Thus recovery is precisely
the event that every leaf in the fixed set $\Dexp(q;X)$ survives.
Independence gives $(1-\varepsilon)^{|\Dexp(q;X)|}$. Adding a cache
segment cannot create an exposed leaf. Since $0<1-\varepsilon<1$,
the probability strictly increases iff that set shrinks.
\end{proof}

\paragraph{Scope of the cache model.}
Corollary~\ref{cor:erasure} models storage loss while the cache survives and source support remains valid. Withdrawing source support invalidates dependent certificates under Theorem~\ref{thm:state-soundness}, so recovery in that setting requires a separate model. Theorem~\ref{thm:bg-modules} states the shared-module and complexity results.

\subsection{Dependencies on Background Results}
\label{app:companion-proofs}

\begin{remark}[Imported ingredients]
\label{rem:external-deps}
The deletion-scan argument is given in
Appendix~\ref{app:proof-kernel}. The rate--distortion and
nonconfusable per-message ingredients retain their dependence
on \citet{semrd}; the scheduling scale retains its dependence
on \citet{derdepth}; and the residual-leaf characterization and
module results retain their dependence on \citet{pvcache,pvbench}.
The concentration step uses \citet{hoeffding}, and allocation
uses the knapsack-constrained submodular result \cite{sviridenko2004}.
\end{remark}

\section{Measurement Protocols and Domain Instances}
\label{app:protocols}

\textbf{The core as a state machine.}
\label{def:core-sm}
The active certified source set $S_t$ is indexed by source records;
the runtime state contains its kernel facts $A_t$, rules $\mathcal R_t$,
certified cache, certificate log, and frequency estimates. $\Admit(u)$
uses the source mode of Algorithm~\ref{alg:gate} and recomputes the
kernel. $\Query(q)$ checks the derived-mode gate or reuses a certificate
validated for the current kernel version. $\Erase(\ell)$ removes the
source record and all cache entries with that transitive dependency;
the kernel is recomputed and surviving cache certificates revalidated
within $\delta$. $\Evict(u)$ removes a cache entry and its dependents
(source eviction follows $\Erase$); $\Tick$ updates frequencies only.
These transitions satisfy Theorem~\ref{thm:state-soundness}, proved in
Appendix~\ref{app:proof-state}. The insertion guard additionally enforces
$\sigma(\mathrm{Cache}_t)\le\mathcal B$; eviction only decreases this
storage. The state can both gain and lose units.

\subsection{A Served Query, End to End}
\label{subsec:boxed}
Algorithms~\ref{alg:deploy}--\ref{alg:serve}
describe the pipeline. The example below follows one L1-HARD query (family \texttt{prelim495x}; story, grounding, and
verdict taken from the stored logs). Both grounding votes return the
same fact set. The three-rule witness below is reconstructed from those
facts and the fixed rules; each step can be checked against the rule set.

\begin{figure}[!htbp]
\centering
\fbox{\begin{minipage}{0.97\columnwidth}
\footnotesize
\textbf{Query} (story, abridged): Chen and Zhou sign a subscription
form (\emph{ren'gou shu}) agreeing to conclude a formal service
contract within thirty days; at the deadline Chen refuses without
cause. \emph{May Zhou claim damages for breach of the preliminary
contract?}

\smallskip
\textbf{Grounding} (bridge, two separately sampled votes; identical
both times):
\begin{itemize}[leftmargin=1.1em,itemsep=0pt,topsep=1pt]
\item[] \texttt{subscription\_order\_or\_booking\_form(x1)}
\item[] \texttt{agrees\_future\_conclusion(x1)}
\item[] \texttt{determinable\_parties\_and\_subject(x1)}
\item[] \texttt{refuses\_to\_conclude\_main\_contract(p1,x1)}
\end{itemize}

\smallskip
\textbf{Gate}: vote $1 ={}$ vote $2$ permits source and derivation
checking; the supported facts and verified steps below permit
certification. Disagreement produces an explicit abstention.

\smallskip
\textbf{Reconstructed certificate} (each step checkable against the
fixed rule set):
\begin{enumerate}[leftmargin=1.4em,itemsep=0pt,topsep=1pt]
\item \texttt{is\_preliminary\_contract(x1)} $\leftarrow$ facts
      1--3 \hfill [Judicial Interp., art.~6]
\item \texttt{fails\_conclusion\_duty(p1,x1)} $\leftarrow$ step 1 +
      fact 4 \hfill [Judicial Interp., art.~7]
\item \texttt{may\_claim\_preliminary\_breach\_liability(p1,x1)}
      $\leftarrow$ steps 1 + 2 \hfill [Civil Code, art.~495]
\end{enumerate}

\smallskip
\textbf{Served answer}: yes (\emph{gold: yes}); derivation depth $3$
$\le$ budget. This record contains no cache event.
\textbf{Verbalization} (carrier, non-critical): ``Yes. The
subscription form constitutes a preliminary contract under Civil Code
art.~495, and Chen's unjustified refusal breaches it.''
\end{minipage}}
\caption{One L1-HARD query with a witness reconstructed from recorded
groundings and fixed rules. Identifiers are shortened. The witness
was reconstructed offline from the facts and outputs stored by the task adapter.}
\label{fig:boxed}
\end{figure}

\begin{algorithm}[!htbp]
\caption{\textsc{Deploy}: theory-grounded instance calibration}
\label{alg:deploy}

\begin{algorithmic}[1]
\Require formalized corpus $S_O$ with rule set $\mathcal R$; checker
         $V$; canonical order $\prec$; query log or surrogate $P_O$;
         hardware cost ratio $\rho$
\Ensure CPUNeSy instance $(A,\delta,\fc,\mathcal B)$; eligibility
        verdict against $\tau$
\State \textbf{D1 Fragment Validation}: re-check every unit and rule
       with $V$; estimate $\hat\epsfr$
\State \textbf{D2/D3 Kernel Extraction}: deletion scan over $\prec$
       with the certified matcher (Alg.~\ref{alg:kernel}, lines 2--3)
\State \textbf{D2/D3} (cont.): order-sensitivity report
       $|A\triangle\Ess(S_O)|$
\State \textbf{D4 Criterion Measurement}: estimate $\widehat{1-P_A}$
       (uniform and weighted, with CIs); verdict vs.\ $\tau$
\State \textbf{D5--D7 Instance Configuration}: $\delta\gets$ 95th
       $P_O$-percentile of $\widehat F_{\Dd}$
\State \textbf{D5--D7} (cont.):
       $\fc(u)\gets\rho\log\big(m_{\mathrm{eff}}+\Ddc(u\mid B)\big)$
       for all $u$
\State \textbf{D5--D7} (cont.): $\mathcal B$ sized around $|A|$; check
       the diminishing-returns condition on the dependency DAG
\State \Return instance configuration and criterion verdicts C1--C4
\end{algorithmic}
\end{algorithm}

\begin{algorithm}[!htbp]
\caption{\textsc{Serve}: online certified query serving}
\label{alg:serve}

\begin{algorithmic}[1]
\Require running core $(A_t,\mathcal R_t,\mathrm{Cache}_t)$; query $q$;
         $k$ carrier traces; depth budget $\delta$
\Ensure certified answer, or explicit abstention with any fallback flagged uncertified
\State \textbf{S1 Grounding Bridge}: ground the traces via interface
       $\mathcal I$; decline if the required agreement fails
\State \textbf{S2 Source Check}: check agreed query facts against their
       source evidence using \textsc{Source} mode of Algorithm~\ref{alg:gate};
       scope these facts to the query
\State \textbf{S3 Derivation Check}: use \textsc{Derived} mode to check
       a complete witness for $q$ over the current checked premises
\If{the gate returns \texttt{CERTIFY}}
  \State \textbf{S4}: serve the certified answer
  \State \textbf{S5 Schedule}: consider its witness for caching when
       $f(q)>\fc(q)$
  \State \textbf{S6 Allocate}: select among eligible witnesses under
       $\sigma(X)\le\mathcal B$ by partial-enumeration greedy
\Else
  \State decline explicitly; label any optional fallback uncertified
\EndIf
\State \textbf{S7 Premise Amendment}: on source withdrawal or expiry of
       query-scoped facts, cascade-evict dependents, recompute the kernel,
       and revalidate surviving cache certificates
\end{algorithmic}
\end{algorithm}

\subsection{Measurement Protocol}
\label{subsec:pipeline}
Both domains use the batch kernel-extraction procedure in Algorithm~\ref{alg:kernel}.

\begin{algorithm}[!htbp]
\caption{Batch kernel extraction and criterion measurement}
\label{alg:kernel}
\begin{algorithmic}[1]
\Require formalized corpus $S_O$ (units + rule set $\mathcal R$);
         canonical order $\prec$; checker $V$; query log or surrogate
         frequency model $P_O$
\Ensure kernel $A$, estimates $\widehat{1-P_A}$, depth distribution
        $\widehat F_{\Dd}$, criterion verdicts C1--C4
\State \textbf{Validate fragment} (C1): re-check every corpus edge with
       $V$; report the fraction of units that fail verification
       (formalization-failure estimate $\hat\epsfr$)
\State \textbf{Deletion scan}: process $S_O$ in order $\prec$; mark $s$
       redundant iff $s\in\Cn(\Gamma\setminus\{s\})$ for the surviving
       set $\Gamma$ (Definition~\ref{def:kernel})
\State $A\gets$ surviving set; $J\gets S_O\setminus A$
\State \textbf{Order-robustness check}: compute
       $\Ess(S_O)=\{s: s\notin\Cn(S_O\setminus\{s\})\}$; report the
       symmetric difference $|A\mathbin{\triangle}\Ess(S_O)|$
\State \textbf{Redundancy estimate} (C3): $\widehat{1-P_A}$ under both
       readings (uniform, $1-|A|/|S_O|$, and frequency-weighted,
       $1-P_O(A)$), with Wilson 95\% CIs over a bootstrap of corpus
       slices
\State \textbf{Depth distribution} (C4): compute $\Dd(s\mid A)$ for all
       $s\in\Cn(A)$ within a cap $\delta_{\max}$ via the predecessor
       operator (Definition~\ref{def:depth}); report quantiles and the
       budget $\delta$ covering 95\% of $P_O$-mass
\State \textbf{Frequency span}: verify that the query-frequency
       distribution straddles $\fc$ over the observed depth range
\State \Return $(A, \widehat{1-P_A}, \widehat F_{\Dd}, \delta)$ and the
       four verdicts
\end{algorithmic}
\end{algorithm}

\subsection{Domain Instance 1: Formal Mathematics}
\label{subsec:math}
\textbf{Corpus.} A fixed snapshot of the Isabelle/HOL standard library
plus selected Archive of Formal Proofs (AFP) entries (alternative:
Lean4 mathlib). Statements (theorems, definitions) become predicate
units; proof-term dependencies become rule edges: lemma $\varphi$ proved
from lemmas $\psi_1,\dots,\psi_k$ by a named inference rule yields the
Horn clause $\psi_1\wedge\cdots\wedge\psi_k\to\varphi$. The fragment is
function-free Horn over a finite statement universe: C1 holds by
construction. The source libraries were checked by their proof kernels. Our executor performs closure queries on their recorded dependency DAGs.

\textbf{C2 (closure fidelity).} Premise selection, theorem retrieval, and automated proof search depend on entailment. Equivalent formulations of a lemma are interchangeable for this criterion. C2 holds.

\textbf{C3 (redundancy).} Proof libraries contain substantial derivable
material (corollaries, specialized instances, alias lemmas). The scan of
Algorithm~\ref{alg:kernel} is run per AFP session with the library order
as $\prec$; $P_O$ is taken (a) uniform and (b) proportional to
downstream dependency counts (how often a lemma is used as a premise),
the latter approximating the service distribution. The protocol hypothesized $1-P_A\ge 0.6$ under both readings. A first kernel-verified
run on a 14-entry subset finds $1-\PAu\approx20\%$ (uniform) and
$46.7\%$ with a $\tau$-crossing interval (weighted). See
Table~\ref{tab:kstep}.

\textbf{C4 (depth and frequency).} Derivation depth is the shortest
premise-chain length in the dependency DAG; the frequency surrogate is
the dependency count. AFP-scale graphs have shallow depth
(order $10$--$10^2$) and heavy-tailed usage, so a budget $\delta$ at the
95th percentile is affordable and query frequencies span several orders
of magnitude around $\fc$. The kernel-level run supports C4, with shallow depth (quantiles $0/2/3/4$, max $10$) and heavy-tailed
usage (frequency quantiles $1/4/82$) on real proof-term dependencies.


\subsection{Domain Instance 2: Regulatory / Legal Compliance Knowledge Base}
\label{subsec:legal}
The regulatory corpus merges the Chinese Contract Book with the full
SPC Interpretation on the General Provisions of the Contract Part (Fa
Shi [2023] No.~13). It contains $503$ Horn units over $200$ articles, with no unsafe rules and $110$ source predicates requiring interpretation of the text. C1/C2 hold. Repeated provisions account for $3$ certified redundancy edges (interpretation arts.~6/8 vs.\ statute
art.~495; art.~53 vs.\ art.~565), giving $1-\PAu=0.60\%$ (Wilson
$[0.20,1.74]\%$), weighted $1.24\%$, schema-alias channel $0.60\%$
(combined $1.19\%$). The merged dependency DAG has $140$ internal edges. Depth quantiles remain $0/1/1/1$ (max $2$), so the additional edges increase branching while chain length stays short. \textbf{C3 fails}
($0.60\%\ll\tau$), and bounded multi-step derivation finds zero additional redundancy at $k\le4$ (Appendix~\ref{subsec:kstep}). \textbf{Mixed-corpus audit (protocol-specified
P-M1--P-M3).} Adding $200$ civil judgments ($15{,}859$ sentence-level
units) shows text-level deduplication removes $6.8\%$ while logical kernel extraction removes a further $2.2$ points, $91.7\%$ of it
instance/weakening redundancy invisible to text and concentrated on the
rule layer ($13.1$--$15.4\%$). The measured reduction falls below P-M2's predicted $\ge10$ points corpus-wide. The composition law (Corollary~\ref{cor:composition}) accounts for this result. Case-specific facts make up $77\%$ of the corpus and are essential under the stated representation assumptions. Corpus-wide savings therefore depend on both rule-layer redundancy and the share of rules in the corpus. We report C3 for the rule layer as well as the full corpus.

\subsection{Setting the Threshold \texorpdfstring{$\tau$}{tau}}
\label{subsec:tau}
The internal protocol sets the threshold using the storage ratio in Theorem~\ref{thm:TA}. By the unit-count consequence of Theorem~\ref{thm:TA}, the unit-count saving of the kernel
store is $|S_O|/|A|=1/\PAu$ (uniform reading of the kernel mass). CPUNeSy eligibility
should require the saving to materially exceed the fixed cost of running
the deductive engine and the bridge (the certification surcharge priced
into $\alpha$, Proposition~\ref{prop:scheduling}). The internal protocol specifies
\[
\tau \;=\; 0.5,
\qquad\text{i.e.\ require}\qquad
|S_O|/|A|\;\ge\;2,
\]
so that an eligible core at least halves the unit count; domains in the
band $1-P_A\in[\tau/2,\tau)$ are declared \emph{hybrid-eligible} and
served with the graceful-degradation protocol
(the main-text hybrid band). A sensitivity analysis over
$\tau\in\{0.3,0.5,0.7\}$ is reported with the results.

\subsection{Threats to Validity}
\label{subsec:threats}
\begin{itemize}
  \item \emph{Order sensitivity.} The scan kernel depends on $\prec$
        unless $A=\Ess(S_O)$; Algorithm~\ref{alg:kernel} reports
        $|A\mathbin{\triangle}\Ess(S_O)|$ so that C3 can be interpreted together with its sensitivity to the deletion order.
  \item \emph{Frequency surrogates / formalization bias.} Both
        uniform and weighted readings are reported; LLM-assisted
        formalization duplicates are removed by the gate's redundancy
        branch itself, and pre-gate statistics are reported.
  \item \emph{Fragment coverage.} Excluded open-textured provisions
        are counted; sub-floor coverage is reported as failing C1.
\end{itemize}

\subsection{Bounded-Step Rederivation and Parent-Designation Sensitivity}
\label{subsec:kstep}
We test two sources of sensitivity in the measurements above. Multi-step derivation may reveal redundancy missed by a single-step check. The designated-parent convention in S5 may also change the depth statistics used to set $\delta$.

\paragraph{Bounded $k$-step C3, $k\in\{1,2,3,4\}$.}
On the legal corpus the check is exact: for each rule unit $r$,
instantiate its body with fresh Skolem constants as temporary facts and
run ${\le}k$ semi-naive forward-chaining rounds over the remaining
units inside the same corpus-order deletion scan as Algorithm~\ref{alg:kernel}. A rule $r$ is $k$-redundant iff its head instance is rederived. If the number of derived facts exceeds the computation cap, the rule is retained, giving a conservative redundancy estimate. On the AFP corpus the certified match
edges (ALIAS ${=}\,\alpha$-equivalence; INST ${=}$ one-step
instantiation) compose transitively, so a $k$-hop path to a surviving
unit is a legitimate $k$-step derivation witness, and the scan deletes
a unit iff such a witness exists within $k$ hops.
Table~\ref{tab:kstep} reports both domains.
\begin{table}[!htbp]
\centering
\caption{Bounded $k$-step rederivation C3 (uniform $1-\PAu$, Wilson
$95\%$). $k{=}1$ reproduces the single-step measurements. Redundancy saturates at $k{=}2$ in both domains and remains below $\tau=0.5$ through $k{=}4$.}
\label{tab:kstep}
\footnotesize
\begin{tabular}{@{}lcccc@{}}
\toprule
Domain ($n$) & $k{=}1$ & $k{=}2$ & $k{=}3$ & $k{=}4$ \\
\midrule
AFP subset ($14{,}414$) & $19.97$ $[19.33,20.63]$ &
$37.13$ $[36.35,37.92]$ & $37.13$ & $37.13$ \\
Legal merged ($503$) & $0.60$ $[0.20,1.74]$ & $0.60$ & $0.60$ & $0.60$ \\
\bottomrule
\end{tabular}
\end{table}
The legal corpus shows zero additional redundancy through $k{=}4$, consistent with its short chains (maximum depth $2$). AFP redundancy nearly doubles at $k{=}2$ because the scan can compose matching edges, for example to remove an instance of an alias. The saturated rate remains $12.9$ points below $\tau$. These results establish the C3 outcome for the tested depth limits.

\paragraph{Designated-parent sensitivity (S5).}
The depth assignment of Appendix~\ref{subsec:pipeline} charges
$\Dd(i)=1+\max_p\Dd(p)$ over premise depths; S5's accounting follows one designated premise edge. To test whether the convention
matters, we re-drew the designated parent uniformly at random on every
multi-premise unit ($18/503=3.6\%$ legal; $2{,}497/14{,}414=17.3\%$
AFP) and recomputed depths over $2{,}000$ resamples. The $95$th uniform
percentile of depth is invariant: constant $1$ (legal) and constant $2$
(AFP), standard deviation $0.0$ in both domains; the max-parent rule is
the conservative upper envelope (AFP $3$ vs.\ $2$). The depth
statistics behind $\delta$ and behind C4's long-tail finding are
therefore insensitive to S5's parent designation.

\section{Full Experimental Details}
\label{app:exp}
The experiments measure corpus redundancy, derivation depth, grounding errors, and task performance against neural and hybrid baselines. Table~\ref{tab:disposition} summarizes the results for each hypothesis in the internal protocol and identifies measurements still needed.

\begin{table}[!htbp]
\centering
\caption{Outcomes for the internally recorded hypotheses. The protocol has no independently verified preregistration timestamp. Outstanding measurements are listed with the results.}
\label{tab:disposition}
\scriptsize
\begin{tabularx}{\columnwidth}{@{}l>{\raggedright\arraybackslash\hsize=0.85\hsize}X>{\raggedright\arraybackslash\hsize=1.15\hsize}X@{}}
\toprule
 & Hypothesis & Observed outcome \\
\midrule
H1 & storage $\le\PAu\times$ fact-inflated at equal accuracy &
Observed candidate-pool ratio: leaf-$\Omega$ uses $48.9\%$ of full-pool units at equal
$97.45\%$ served accuracy; physical storage not measured \\
H2 & $\fc$-scheduled serves cheaper than unrestricted certified at
equal served accuracy & $-43\%$ tokens under candidate restriction; scheduler-specific crossover remains unmeasured \\
H3 & certified coverage $\ge$ L-B lower bound & Not empirically instantiated: logged outcome rates differ from the step-level channel rates; independence remains unverified \\
H4 & erasure curves match $(1-\varepsilon)^{\ell}$ within Wilson CIs &
Supported by Monte Carlo simulation \\
H5 & zero reasoning shortcuts by construction & Structural admission replay rejects all 120 injected target writes; semantic leaf errors remain possible \\
H6 & degradation rate $\le$ C1 floor; no silent uncertified answers &
Explicit refusal/status recorded; ContractNLI neutral and fallback outputs are uncertified \\
H7 & kernel store smaller than retrieval index by $1/\PAu$ &
Corpus redundancy measured; matched index-byte comparison deferred \\
H8 & substrate parity at lower memory on M1--M2 & Equal observed served accuracy with a smaller candidate pool; kernel-store byte accounting remains separate \\
H9 & in-fragment accuracy $\ge$ pure/RAG baselines with zero silent
answers (L1) & Supported
(Table~\ref{tab:rq2-ablation}) \\
H10 & every certified answer carries a derivation $\le\delta+\drec$ &
Formal bounded-witness contract; complete per-query certificate archival was not instrumented in every adapter \\
P-M1 & statutory layer stays at $0.60\%$ in mixed corpus &
\textbf{held} \\
P-M2 & $\ge10$-point logical margin corpus-wide & \textbf{falsified as
stated} ($2.2$ points); explained by the composition law
(Corollary~\ref{cor:composition}), which predicts the measured value
within half a point \\
P-M3 & instance/weakening dominates the redundant mass & \textbf{held}
($91.7\%$) \\
\bottomrule
\end{tabularx}
\end{table}

\subsection{LD-MH: Replication on the Real mathlib4 DAG}
\label{subsec:ldmh}
We applied the M-HARD generator to the mathlib4 dependency DAG extracted by the LeanDojo tracer. The generator uses the same depth strata $2$--$6$, near-miss negatives, and twin-extended closure procedure for gold labels. The neural baselines remain near chance, with different errors on positive and negative queries. The certified channel answers nearly all queries with high accuracy. Table~\ref{tab:ldmh} reports the results on this second proof library.

\subsection{External Validity on LeanDojo Benchmark 4}
\label{subsec:leandojo}
We also evaluate on \emph{LeanDojo Benchmark~4} \cite{leandojo}, the official NeurIPS~2023 benchmark extracted from mathlib4 (commit \texttt{29dcec07}). We measure the four applicability criteria, premise selection, and multi-hop closure on this independently constructed corpus.

\textbf{Four criteria on the real corpus.} C1, C2, and C4 hold
directly ($120{,}509$ theorems, $180{,}907$ defined premises;
$63.3\%$ of annotated theorems multi-hop; depth long-tailed
$p_{50}{=}2$, $p_{90}{=}9$, $p_{99}{=}27$). C3 remains unresolved because each traced theorem records only one derivation. The observed $36.5\%$ premise usage measures which premises appear in those traces. Establishing redundancy would require testing whether omitted premises can be derived from the retained set. The AFP and regulatory audits perform such derivability checks.

\textbf{LD-PS: premise selection.} On $300$ theorems sampled from the
official \texttt{random/test} split (gold = premises of the traced
proof; candidate pools of $15$), the reference grounder scores:
pure recall-from-memory
$1.3\%$ exact; BM25-pool RAG $11.3\%$ exact (precision $85.6\%$);
CPUNeSy (kernel-restricted pool, whitelist certificate, 2-vote
abstention) $11.2\%$ exact \emph{on $83\%$ coverage} with
\emph{$98.9\%$ precision} (Table~\ref{tab:ldps}). On the two weaker grounders, CPUNeSy raises answered exact-set accuracy (DeepSeek $17.0\%$ vs.\ RAG $8.7\%$; Qwen $19.9\%$ vs.\ $12.0\%$), at $39.3/57.0\%$ coverage. Full-pool exact-set accuracy is $6.7/11.3\%$, with gains concentrated in the answered subset. The reference run likewise has $9.3\%$ full-pool exact-set accuracy versus BM25's $11.3\%$. Gold premises come from one traced proof among potentially many valid proofs; this limits the interpretation of exact-set matching. The models also select few premises, averaging $0.46$--$0.51$ against a gold mean of $2.76$. This pattern indicates frequent omission of needed premises.

\textbf{LD-MH: closure judgment on the real DAG.} The Seed 2.0 Lite
cross-vendor run uses the same $200$ depth-stratified tasks ($99$ pos
/$101$ neg), hard distractors, and deterministic DAG-closure scorer as
the synthetic protocol. Pure CoT obtains $50.0\%$ and RAG $51.0\%$;
one-vote grounding reaches $94.0\%$, while the two-vote certified
channel obtains $93.9\%$ on its $98.5\%$ served subset and $92.5\%$
full-pool accuracy (Table~\ref{tab:ldmh}). Thus the certificate
abstains on three tasks and preserves a $41.5$-point full-pool
advantage over the better neural control. This run tests the same closure task with a different grounding model from the primary K3 configuration.

LD-MH measures derivability over recorded dependency graphs. At $200$/$300$ tasks, the Wilson half-widths are $\pm 3$--$4$ points. The measured execution advantage transfers to a second proof graph under the same task generator. End-to-end Lean proof search would require a separate evaluation.

\begin{table}[!htbp]
\centering
\caption{LD-PS ($n=300$): full-pool and answered exact-set accuracy, premise precision/recall, and coverage. Abstentions count as errors in the full-pool column. These separate selection quality from the subset selected for answering.}
\label{tab:ldps}
\footnotesize
\begin{tabularx}{\columnwidth}{@{}Xlccccc@{}}
\toprule
Strategy & model & exact (all) & exact (served) & precision & recall & coverage \\
\midrule
pure (recall from memory) & reference & $1.3\%$ & $1.3\%$ & $70.0\%$ & $3.4\%$ & $100\%$ \\
rag (BM25, full corpus) & reference & $11.3\%$ & $11.3\%$ & $85.6\%$ & $14.4\%$ & $100\%$ \\
CPUNeSy (kernel pool + cert. + 2-vote) & reference & $9.3\%$ & $11.2\%$ & $\mathbf{98.9\%}$ & $13.5\%$ & $83.0\%$ \\
\midrule
pure (recall from memory) & DeepSeek & $2.0\%$ & $2.0\%$ & $9.0\%$ & $14.8\%$ & $100\%$ \\
rag & DeepSeek & $8.7\%$ & $8.7\%$ & $45.1\%$ & $22.0\%$ & $100\%$ \\
CPUNeSy & DeepSeek & $6.7\%$ & $\mathbf{17.0\%}$ & $68.0\%$ & $31.4\%$ & $39.3\%$ \\
\midrule
pure & Qwen2.5-72B & $1.3\%$ & $1.3\%$ & $7.1\%$ & $8.9\%$ & $100\%$ \\
rag & Qwen2.5-72B & $12.0\%$ & $12.0\%$ & $48.6\%$ & $34.0\%$ & $100\%$ \\
CPUNeSy & Qwen2.5-72B & $11.3\%$ & $\mathbf{19.9\%}$ & $59.8\%$ & $43.8\%$ & $57.0\%$ \\
\bottomrule
\end{tabularx}
\\[2pt]
{\scriptsize Models select $0.46$--$0.51$ premises on average, compared with a gold mean of $2.76$.}
\end{table}

\textbf{Learned-retriever control (ReProver).}
We evaluate LD-PS candidate pools using the official pretrained ReProver retriever
(byt5-small encoder, trained by its authors on this benchmark's train
split \cite{leandojo}), run locally with no task-specific training.
We use two retrieval protocols. \emph{(i) Native tactic-state protocol}: over $2{,}641$
traced test tactic states against the full $180{,}907$-premise corpus,
mean-pooled encoder embeddings and cosine ranking give per-premise
recall R@$1$ $=8.3\%$, R@$5$ $=19.3\%$, R@$10$ $=26.0\%$,
R@$100$ $=48.1\%$, and full gold-set coverage FC@$10$ $=18.0\%$.
\emph{(ii) LD-PS pool protocol}: a ReProver top-$15$ pool covers all
gold premises for $9.0\%$ of the $300$ anchor tasks when queried by the
theorem statement (the same input the LLM arms saw; recall@$15$
$=18.5\%$), and $17.7\%$ when queried by the first tactic state, which provides information unavailable to the LLM arms (recall@$15$ $=31.5\%$).
The BM25 and kernel-restricted pools both stand at $76.0\%$ full
coverage and $89.3\%$ recall@$15$. BM25 and the kernel-restricted pool have higher coverage than ReProver in this candidate-pool evaluation. Their equal coverage makes BM25 the stronger retrieval control here. A storage comparison would also require measuring matched index sizes. ReProver normally ranks premises throughout proof search, whereas our tactic-state protocol evaluates only the first state. Larger or fine-tuned retrievers and complete proof-search runs remain to be tested.

\begin{table}[!htbp]
\centering
\caption{LD-MH (real mathlib4 DAG mirror of M-HARD, $n{=}200$).
Seed 2.0 Lite is a completed cross-vendor control.
Cells report full-pool accuracy except the coverage column; the
Seed two-vote certified channel is $93.9\%$ accurate on its served subset.}
\label{tab:ldmh}
\footnotesize
\begin{tabularx}{\columnwidth}{@{}Xccccc@{}}
\toprule
 & pure & rag & 1-vote & CPUNeSy & cov. \\
\midrule
M-HARD (synthetic AFP, K3) & $53.5$ & $55.0$ & $\mathbf{96.5}$ & $95.5$ & $98.0$ \\
LD-MH (real mathlib4, Seed 2.0 Lite) & $50.0$ & $51.0$ & $\mathbf{94.0}$ & $92.5$ & $98.5$ \\
LD-MH (real mathlib4, DeepSeek) & $51.5$ & $50.0$ & $\mathbf{83.5}$ & $79.0$ & $89.5$ \\
LD-MH (real mathlib4, Qwen2.5-72B) & $52.0$ & $52.5$ & $\mathbf{77.5}$ & $71.5$ & $91.5$ \\
\bottomrule
\end{tabularx}
\\[2pt]
{\scriptsize The 1-vote column is the deterministic-grounding ablation; CPUNeSy is the two-vote policy. High-budget M-HARD CoT and RAG score $60.5\%$ ($121/200$) and $59.5\%$ ($119/200$), respectively, so the margin over the stronger high-budget control is $35.0$ points. The paired counts appear in Table~\ref{tab:paired-discordants}. DeepSeek one-vote and coverage cells use the same frozen run (Appendix~\ref{app:gate-audit}).}
\end{table}

\subsection{MedCalc-Bench Verified: Protocol Notes}
\label{subsec:medcalc-protocol}
The pure baseline uses the official zero-shot CoT protocol. The open-book baseline adds calculator formulas to the context. CPUNeSy uses the model named in each column to extract entities, executes the verified reference implementations, and abstains on extraction disagreement (one vote) or disagreement between interfaces (two votes). The fallback cascade routes abstentions to the same model's open-book answer. All $55$ reference implementations reproduce $100\%$ of gold answers from gold entities, so residual errors come from the extracted inputs. Table~\ref{tab:rq1-ledger} reports the main accuracy table for Seed
and DeepSeek, with native two-vote full-pool accuracies of $72.5\%$
and $63.5\%$, respectively. For Seed 2.0 Lite, the
one-vote ablation serves $969/1{,}100$ items at $89.37\%$ accuracy
($78.73\%$ full-pool); the two-vote channel serves $858/1{,}100$ at
$92.89\%$ accuracy ($72.45\%$ full-pool, rounded to $72.5\%$ in the
table). The updated DeepSeek result is reported at full-pool level;
its two-vote channel serves $733/1{,}100$ items, with $698$ correct,
for $95.23\%$ answered accuracy, $66.64\%$ coverage, and $63.45\%$
full-pool accuracy (rounded to $63.5\%$ in the table). The recorded
result file retains item-level records across all $55$ calculators for
stratified analysis.

\paragraph{Supplementary GPT-5.2 run.}
A supplementary run used GPT-5.2 through an aggregator API. With no fallback,
its two-vote channel reaches $91.16\%$ answered accuracy at $72.0\%$
coverage ($65.64\%$ full-pool). The Seed--DeepSeek comparison uses the model profiles in Table~\ref{tab:model-identities}.

\subsection{Fairness and Robustness Controls (Full)}
\label{app:fairness}
The internal protocol records hypotheses H1--H10 and P-M1--P-M3 with falsification thresholds (Table~\ref{tab:disposition}). The task generators are machine-checked, and Appendix~\ref{app:gate-audit} describes protocol provenance.

A weaker grounding model retains a ${\sim}36$-point advantage on the compositional stratum against frontier baselines using ${\le}3\times$ budget. Regenerating each suite with $3$ seeds preserves the reported performance gaps. The gaps also persist under query paraphrases and across Jaccard-ranked top-1, top-5, and random distractor pools.

We report Wilson $95\%$ intervals for rates and exact McNemar tests for paired main comparisons. The uncertified-solver comparison gives a two-sided $p{=}0.34$ ($b{=}3$, $c{=}7$).

\subsection{Observed Error Counts and Channel-Model Scope}
\label{subsec:channel}
The recorded error counts characterize observed outcomes at several levels. For L1-HARD and LD-MH, respectively, the logs report $20/1{,}283$ and $34/1{,}853$ grounding-item failures; $0/457$ and $15/301$ wrong verdicts conditional on clean grounding; $20/477$ and $27/315$ incorrect accepted traces; and $9/235$ and $11/192$ wrong agreed outputs in the broad-interface two-vote runs. These fractions have different denominators and describe different events. The $20$ L1-HARD trace errors include assertions of the queried conclusion as a premise. Separately, the restricted leaf-interface policy yields no observed wrong answers among $233$ served cases; changing interfaces also changes the proposals supplied to voting (Appendix~\ref{app:gate-audit}). The reported LD-MH errors conditional on clean grounding concentrate on negatives with intermediate-node routes around the broken leaf in the recorded dependency graph.

These outcome rates do not directly estimate the parameters of Proposition~\ref{prop:coverage}. In Definition~\ref{def:bridge-channel}, $\mu$ bounds loss or corruption of an intended derivation step, and $\eta$ bounds rejection of a \emph{valid} step. An incorrect accepted trace is a different event from a valid step being rejected. Likewise, an aggregate grounding-item failure fraction is not by itself a uniform upper bound on each exposed premise's formalization failure. We therefore report the counts as descriptive error diagnostics and do not substitute them into the coverage bound. Empirical instantiation would require annotations of the corresponding step-level events, appropriate probability bounds, and assessment of the stated independence assumptions. The proposition remains a conditional guarantee.

\section{Admission Audit and Reproducible Offline Analysis}
\label{app:gate-audit}

\subsection{Source checking and fixed-trace gate effect}
ContractNLI uses a learned rechecker on the source document and each contributing fact. A fact must receive \texttt{holds}; otherwise the vote abstains. Conflicting closure results also abstain. Two agreeing entailment/contradiction votes set the certificate flag; two neutral votes produce the separate \texttt{neutral\_uncertified} status. The main native-policy row includes this explicitly uncertified neutral status, whereas strict certificate-only statistics require the flag.

The gate ablation compares raw two-vote agreement with the stored output after rechecking, using the same grounding traces. Seed changes $(C,W,A)$ from $(858,128,51)$ to $(844,120,73)$; DeepSeek changes from $(1533,252,306)$ to $(1482,189,420)$. Rechecking withholds answers while leaving all other outputs unchanged. The counts measure which correct and wrong answers are withheld. Estimating premise-level false acceptance and rejection would require independent source-support annotations, which these records lack.

MedCalc parses calculator-specific fields, normalizes units and values, requires agreement, and calls the deterministic calculator. The recorded one-vote and native two-vote policies give $(C,W,A)=(866,103,131)\to(797,61,242)$ for Seed and $(824,88,188)\to(698,35,367)$ for DeepSeek. This comparison includes the agreement requirement and is distinct from the ContractNLI source-recheck ablation. The agreed, schema-valid attributes can still misrepresent the clinical note. The risk estimates include the resulting wrong answers.

\subsection{L1-HARD replay and structural admission}
The historical leaf interface specifies the allowed vocabulary in the prompt. Its executor accepts any parsed predicate. We therefore provide a separate explicit admission replay that rejects predicates outside the rule-derived leaf vocabulary. All parsed proposals in the recorded restricted runs use leaf predicates. The offline structural stress test uses generator facts restricted to allowed leaves as clean controls, then adds the target atom to each negative item. The barrier rejects $120/120$ injected target writes in this constructed test. The raw solver accepts every injected target by construction. Generator facts supply the clean replay controls and scoring labels. The historical model-serving runs obtain their proposed facts from the model.

The broad two-vote joint outcomes are $226$ both correct, $9$ both wrong, $2$ correct/wrong, and $3$ pairs with one parsing failure. The leaf two-vote outcomes are $233$ both correct, $2$ correct/wrong, $2$ wrong/correct, and $3$ pairs with one parsing failure. Thus the historical $20/477$ single-trace errors belong to the broad interface, while the $0/233$ served-error observation belongs to the leaf interface. Comparing these runs changes both the vocabulary and the proposals supplied to the voting rule. Against exact generator source-fact strings, the $1{,}804$ distinct-per-trace leaf proposals contain $1{,}732$ matches and $72$ mismatches. These are exact string matches. Assessing whether the mismatches are semantically supported would require human review of the source.

\subsection{Statistics, latency, and calibration}
All policies share their task denominator; parse failures count as abstentions. The L1 accuracy-difference interval uses $20{,}000$ paired item bootstrap resamples with seed $20260926$. The $0/233$ upper error bound is $1-0.05^{1/233}$ under an item-binomial sampling model. The comparisons are exploratory and report counts and confidence intervals. Table~\ref{tab:paired-discordants} gives Holm-adjusted tests for the six paired comparisons. Two traces are insufficient to estimate the theorem's latent $\bans$ and $\slipc$ separately.

A prospective deployment protocol would freeze a disjoint calibration set with source annotations, label both-wrong, discordant, and invalid traces, measure accepted-answer risk and coverage, choose $c_W,c_A,c_T$ before testing, then freeze the write vocabulary and serving rule on a separate test set. Our experiments report the recorded joint outcomes for retrospective analysis. Testing the proposed deployment rule on held-out data remains future work. We calculate latency by summing logged call durations. The logs lack complete measurements of deployment formalization, expert review, concurrency, and monetary cost.

\subsection{Protocol provenance}
The artifact contains a dated internal mixed-corpus protocol and a record of its hypotheses. There is no independent timestamp establishing preregistration. Scripts and input hashes make the present analyses reproducible. Table~\ref{tab:disposition} reports the outcome of each hypothesis and the measurements still needed.

\subsection{Supplementary gate effect and cost analysis}
\begin{figure}[!htbp]
\centering
\includegraphics[width=\linewidth]{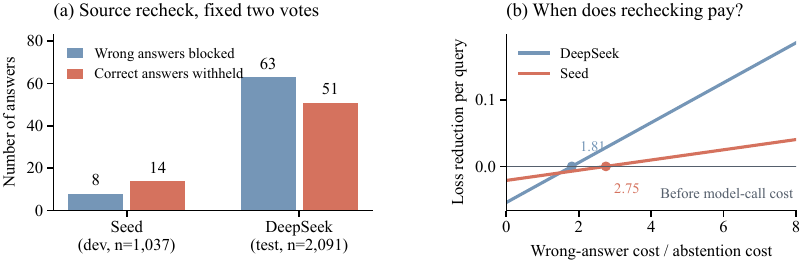}
\caption{Additional fixed-trace source-recheck analysis on ContractNLI. The left panel counts correct and wrong answers withheld. The right panel uses these counts to calculate loss reduction at different relative error costs.}
\label{fig:gate-specific}
\end{figure}
For policies on the same $N$ items, let $c_W,c_A,c_T$ denote the respective costs of a wrong answer, refusal, and token. The loss difference is
\begin{equation}
\Delta L=(\Delta W c_W+\Delta A c_A+\Delta T c_T)/N.
\label{eq:decision-cost}
\end{equation}
Before compute cost, ContractNLI source rechecking improves this loss when $c_W/c_A>114/63=1.81$ for DeepSeek and $>22/8=2.75$ for Seed. Restricted L1-HARD two-call serving is preferred when $2c_W>6c_A+145{,}422c_T$ across $240$ items. The paired accuracy-difference interval is $[-4.17,0.83]$ percentage points (two-call minus single-call), which does not establish equivalence. Median sums of logged call durations are $6.97$ versus $13.88$ seconds. These measurements assume sequential calls.

\subsection{Selective operating points and paired tests}
\label{app:selective-points}
Table~\ref{tab:selective-points} reports discrete serving policies, with wrong answers divided by served answers for risk and served answers divided by the common pool for coverage. Each row is a discrete policy. Estimating a risk--coverage curve or its area would require a confidence-threshold sweep. Wilson intervals use the item-level sampling model.
\begin{table}[!htbp]\centering\scriptsize
\caption{Matched error--coverage accounting from frozen records. C/W/A denote correct/wrong/abstained counts. Risk and coverage are percentages; brackets are 95\% Wilson intervals for answered risk.}
\label{tab:selective-points}
\begin{tabular}{@{}llrrrrr@{}}\toprule
Data / model & Policy & C & W & A & Coverage & Risk [95\% CI]\\\midrule
ContractNLI Seed & Agreement only & 858 & 128 & 51 & 95.1 & 13.0 [11.0, 15.2]\\
 & Native policy & 844 & 120 & 73 & 93.0 & 12.4 [10.5, 14.7]\\
 & Certificate flag only & 513 & 75 & 449 & 56.7 & 12.8 [10.3, 15.7]\\
 & + fallback & 880 & 157 & 0 & 100.0 & 15.1 [13.1, 17.4]\\
\addlinespace[2pt]
ContractNLI DeepSeek & Agreement only & 1533 & 252 & 306 & 85.4 & 14.1 [12.6, 15.8]\\
 & Native policy & 1482 & 189 & 420 & 79.9 & 11.3 [9.9, 12.9]\\
 & Certificate flag only & 872 & 119 & 1100 & 47.4 & 12.0 [10.1, 14.2]\\
 & + fallback & 1742 & 349 & 0 & 100.0 & 16.7 [15.2, 18.3]\\
\addlinespace[2pt]
MedCalc Seed & One vote & 866 & 103 & 131 & 88.1 & 10.6 [8.8, 12.7]\\
 & Native policy & 797 & 61 & 242 & 78.0 & 7.1 [5.6, 9.0]\\
 & + fallback & 991 & 109 & 0 & 100.0 & 9.9 [8.3, 11.8]\\
\addlinespace[2pt]
MedCalc DeepSeek & One vote & 824 & 88 & 188 & 82.9 & 9.6 [7.9, 11.7]\\
 & Native policy & 698 & 35 & 367 & 66.6 & 4.8 [3.5, 6.6]\\
 & + fallback & 994 & 106 & 0 & 100.0 & 9.6 [8.0, 11.5]\\
\addlinespace[2pt]
LD-MH Seed & One vote & 188 & 12 & 0 & 100.0 & 6.0 [3.5, 10.2]\\
 & Native policy & 185 & 12 & 3 & 98.5 & 6.1 [3.5, 10.3]\\
 & + fallback & 185 & 15 & 0 & 100.0 & 7.5 [4.6, 12.0]\\
\addlinespace[2pt]
LD-MH DeepSeek & One vote & 167 & 33 & 0 & 100.0 & 16.5 [12.0, 22.3]\\
 & Native policy & 158 & 21 & 21 & 89.5 & 11.7 [7.8, 17.3]\\
 & + fallback & 167 & 33 & 0 & 100.0 & 16.5 [12.0, 22.3]\\
\addlinespace[2pt]
\bottomrule\end{tabular}\end{table}
\begin{table}[!htbp]\centering\small
\caption{Paired full-pool accuracy tests. $b$ counts CPUNeSy-only correct items; $c$ counts baseline-only correct items. Abstentions count as incorrect. Holm correction is applied jointly to these six exploratory comparisons.}
\label{tab:paired-discordants}
\begin{tabular}{@{}lrrrr@{}}\toprule Comparison & $b$ & $c$ & Exact $p$ & Holm $p$\\\midrule
M-HARD vs pure & 92 & 8 & $3.21\times10^{-19}$ & $1.60\times10^{-18}$\\
M-HARD vs rag & 87 & 6 & $1.65\times10^{-19}$ & $9.90\times10^{-19}$\\
M-HARD vs pure\_hi & 78 & 8 & $1.52\times10^{-15}$ & $4.57\times10^{-15}$\\
M-HARD vs rag\_hi & 77 & 5 & $1.20\times10^{-17}$ & $4.82\times10^{-17}$\\
L1-HARD vs pure & 36 & 6 & $2.83\times10^{-6}$ & $5.66\times10^{-6}$\\
L1-HARD vs rag & 28 & 6 & $1.95\times10^{-4}$ & $1.95\times10^{-4}$\\
\bottomrule\end{tabular}\end{table}

\subsection{Reproducible artifact and recorded model identities}
The accompanying offline audit package contains the frozen input records, source-file SHA-256 hashes, scripts, and generated counts. Its README gives commands that require no API credentials or network calls. Table~\ref{tab:model-identities} lists the available model identifiers. The caches lack consistent records of sampling dates and immutable model revisions.
\begin{table}[!htbp]\centering\footnotesize
\caption{Recorded model names and their scope. Primary public columns use the profiles identified here; historical adapter runs are a separate campaign.}
\label{tab:model-identities}
\begin{tabularx}{\linewidth}{@{}lXX@{}}\toprule
Paper label & Recorded identifier & Scope / provider record\\\midrule
K3 / reference & \texttt{k3-agent} & Controlled L1-HARD and M-HARD; model profile configured by experiment environment\\
K2.6 / weaker same-family & \texttt{k2d6-agent} & Same-family control; model profile configured by experiment environment\\
DS (L1 controls) & \texttt{deepseek-chat} & Cross-model L1 study; distinct from public-column Flash\\
QW & \texttt{Qwen/Qwen2.5-72B-Instruct} & Cross-model control accessed through a public API\\
GLM & GLM-5.2 & Cross-model control; provider label in the experiment record\\
Seed (public columns) & Seed 2.0 Lite (profile identifier redacted) & Recorded profile Seed 2.0 Lite\\
DeepSeek (public columns) & \texttt{deepseek-flash} & Official DeepSeek API profile\\
GPT-5.2 & GPT-5.2 aggregator run & Supplementary MedCalc only\\\bottomrule
\end{tabularx}\end{table}
Two-vote majority with abstention on ties is identical to the binary agreement rule on the same valid traces. The fixed-trace ablation therefore measures the additional effect of source rechecking after this voting rule. The archive preserves all policy counts, including incorrect accepted answers and explicitly uncertified neutral/fallback statuses.

\subsection{Process-figure provenance}
\label{app:figure-provenance}
Figures~\ref{fig:reasoning-state-trajectories}--\ref{fig:grounding-geometry} use K3 for L1-HARD and LD-MH, and Seed for ContractNLI and MedCalc. The LD-MH geometry panel shows the $181/19$ split from K3, while the public accuracy table reports the $185/15$ split from Seed. The plotted features summarize grounding, admission, closure, and output logs. Principal component analysis (PCA) and adjusted Rand index (ARI) describe where the recorded protocols diverge. The fixed-trace comparison in Section~\ref{subsec:rq2} measures the effect of source rechecking.

\subsection{Historical Task-Adapter Results}
\label{app:historical-adapters}
Table~\ref{tab:historical-adapters} preserves the historical task-adapter summaries. These runs use other splits and model profiles than the primary matched comparisons. Their per-row split identifiers, model profiles, and sample counts are not specified in the available summaries, so the archived column labels do not establish matching to Table~\ref{tab:rq1-ledger}. We do not use these scores to rank the primary CPUNeSy policies against the original systems. Near-zero scores also require a failure breakdown: the summaries do not separate parsing failures, solver failures, and wrong executable outputs. The aggregate scores alone cannot determine their cause.

\begin{table}[!htbp]
\centering\scriptsize
\caption{Historical task-adapter full-pool accuracy (\%). Column names reproduce archive labels, not matched model/split identities. Exact per-row denominators are unspecified in these summaries; these results are descriptive and excluded from primary comparative claims.}
\label{tab:historical-adapters}
\begin{tabular}{@{}lcccccc@{}}
\toprule
& \multicolumn{2}{c}{ContractNLI} & \multicolumn{2}{c}{MedCalc} & \multicolumn{2}{c}{LD-MH}\\
\cmidrule(lr){2-3}\cmidrule(lr){4-5}\cmidrule(l){6-7}
Adapter & Seed & DeepSeek & Seed & DeepSeek & Seed & DeepSeek\\
\midrule
Logic-LM \citep{logiclm} & $0.0$ & $0.0$
         & $75.5$ & $77.5$
         & $50.5$ & $46.5$ \\
DetermLR \citep{sun2024determlr} & $48.4$ & $70.0$
         & $67.5$ & $61.5$
         & $50.5$ & $49.5$ \\
LINC \citep{olausson2023linc} & $16.0$ & $0.2$
     & $66.5$ & $44.7$
     & $49.5$ & $51.0$ \\
SymbCoT \citep{xu2024symbcot} & $56.8$ & $16.5$
        & $66.5$ & $75.1$
        & $49.5$ & $49.5$ \\
LINA \citep{li2024lina} & $48.5$ & $10.8$
     & $67.3$ & $69.3$
     & $49.5$ & $51.0$ \\
HBLR \citep{li2025hblr} & $54.7$ & $10.4$
     & $66.3$ & $62.8$
     & $49.0$ & $50.5$ \\
MenTaL \citep{mental} & $1.8$ & $0.0$
       & $75.0$ & $76.2$
       & $47.5$ & $50.0$ \\
\bottomrule
\end{tabular}
\end{table}
\begingroup
\setlength{\intextsep}{8pt plus 1pt minus 1pt}
\section{Supplementary Visualizations}
\label{app:visuals}
Figures~\ref{fig:pareto}, \ref{fig:crossover}, and~\ref{fig:deltavisual} illustrate the theoretical models with schematic curves. The other panels report stress-test results or Monte Carlo erasure trials.

\begin{figure}[H]
\centering
\includegraphics[width=.80\linewidth]{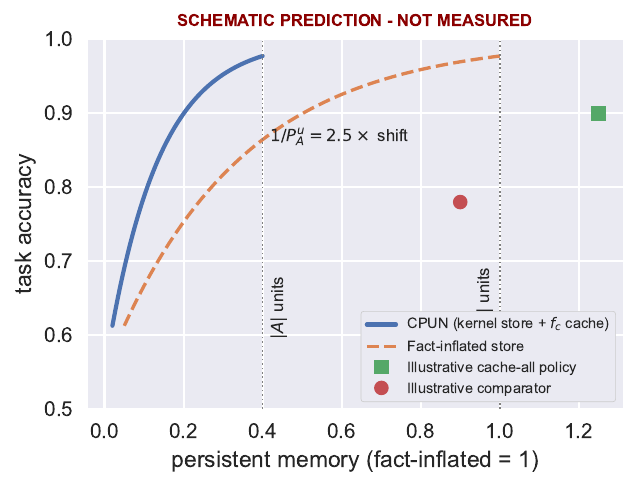}
\caption{Illustrative memory--accuracy prediction under the kernel-storage
model. The curves and reference points are schematic. The legacy legend label CPUN denotes CPUNeSy.}
\label{fig:pareto}
\end{figure}

\begin{figure}[H]
\centering
\includegraphics[width=\linewidth]{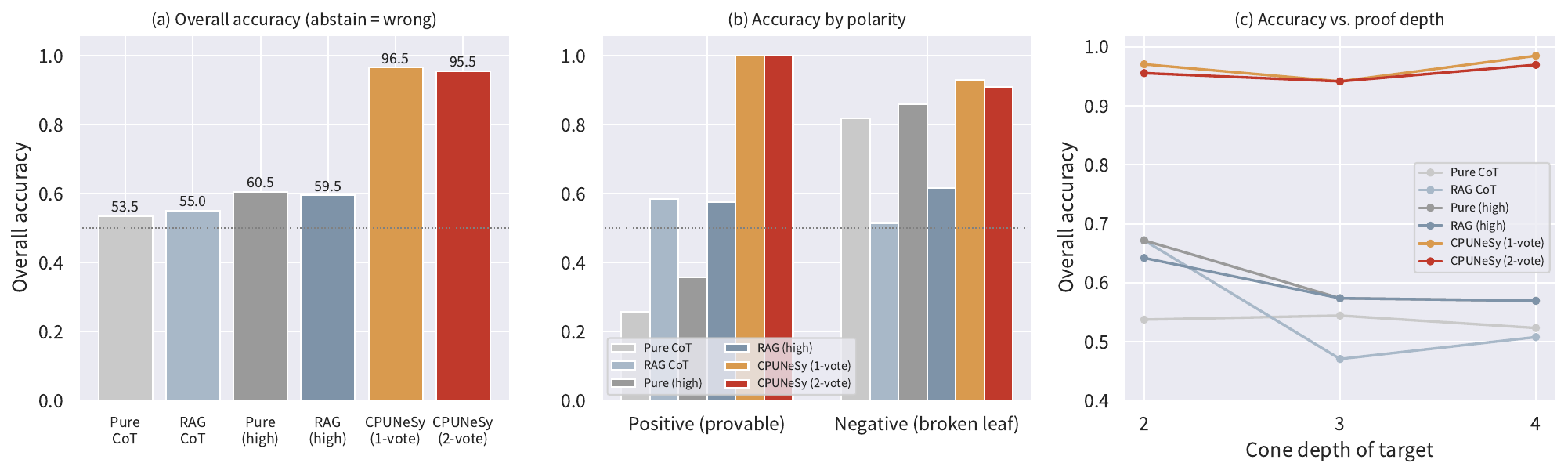}
\caption{M-HARD visual summary. CPUNeSy separates from pure and
retrieval baselines in overall accuracy, removes the polarity asymmetry
on positives versus near-miss negatives, and stays stable as proof depth
increases.}
\label{fig:mhardvisual}
\end{figure}

\begin{figure}[H]
\centering
\includegraphics[width=.96\linewidth]{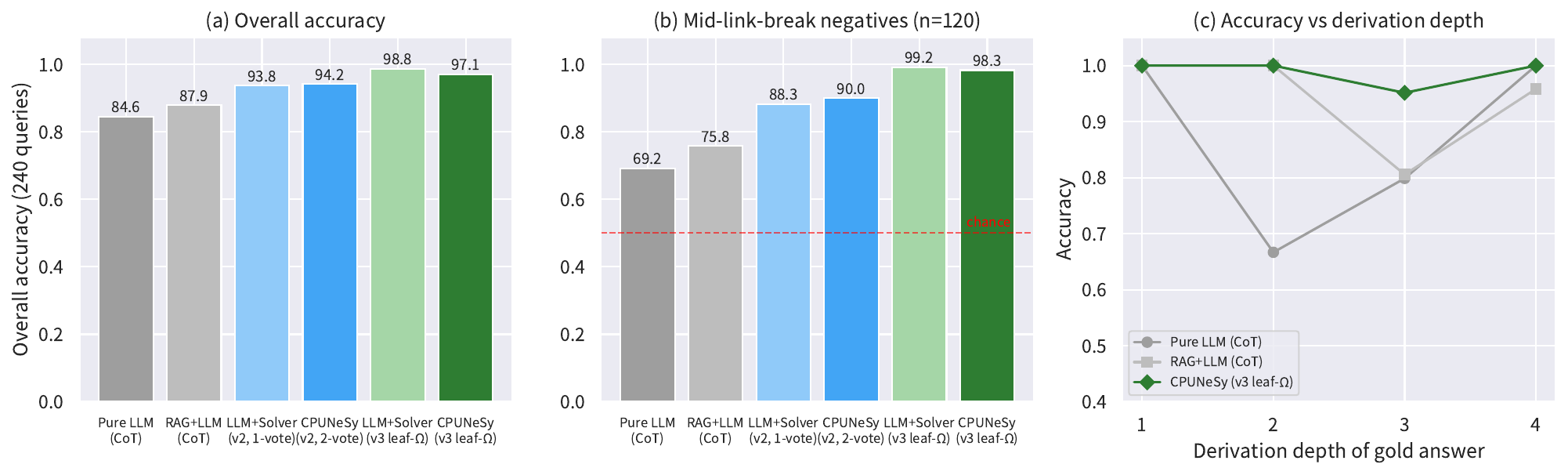}
\caption{L1-HARD visual summary. CPUNeSy's accuracy gain is concentrated
on mid-link-break negatives, where factual access alone is insufficient.
Section~\ref{subsec:rq2} measures each component's contribution through ablation.}
\label{fig:l1hardvisual}
\end{figure}

\begin{figure}[H]
\centering
\includegraphics[width=.55\linewidth]{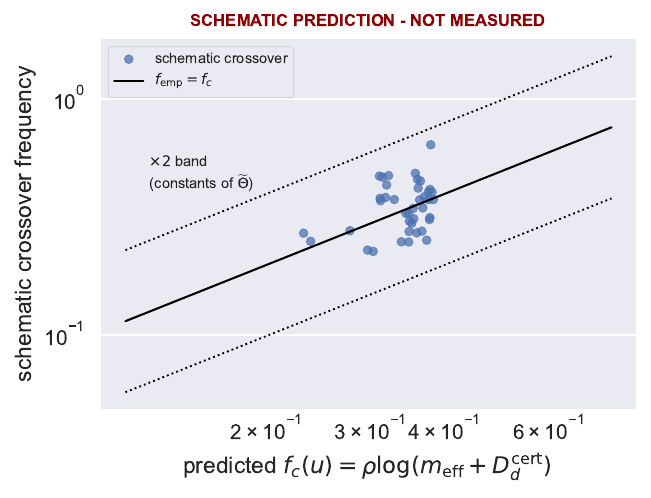}
\caption{Illustrative crossover prediction under the stationary
store--compute model. The schematic points and constant-factor band illustrate the predicted relationship between the caching threshold and crossover frequency.}
\label{fig:crossover}
\end{figure}


\begin{figure}[H]
\centering
\includegraphics[width=.94\linewidth]{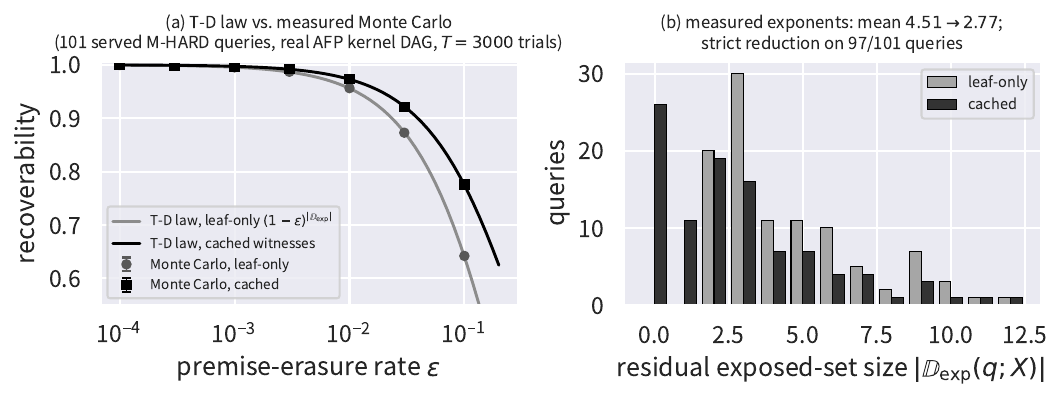}
\caption{Monte Carlo premise-erasure recovery on 101 served M-HARD
queries and their AFP witness DAGs ($3{,}000$ trials). The left panel compares recovery with the residual-leaf law. The right reports exposed-set sizes before and after caching. T-D denotes Corollary~\ref{cor:erasure}. The simulation assumes storage loss with a surviving cache and unchanged source support.}
\label{fig:erasurevisual}
\end{figure}

\begin{figure}[H]
\centering
\includegraphics[width=.55\linewidth]{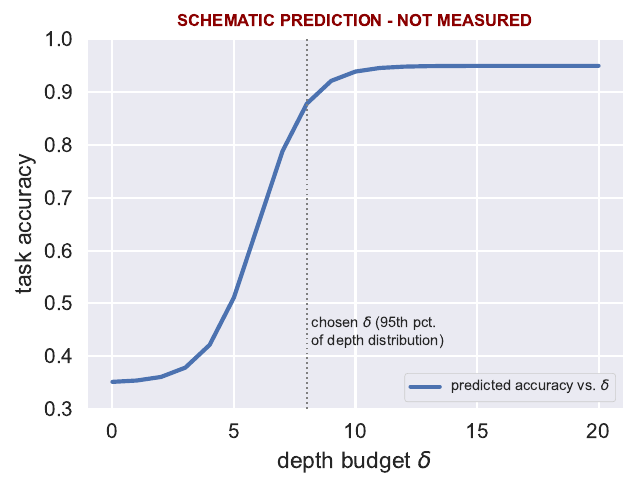}
\caption{Illustrative depth-budget sensitivity prediction. The schematic curve illustrates saturation near a budget chosen from the depth distribution.}
\label{fig:deltavisual}
\end{figure}

\endgroup

\section{Cross-Model Evaluation}
\label{app:boundary}
\subsection{Cross-Model Error Regimes on L1-HARD}
\label{subsec:crossregimes}
Table~\ref{tab:crossckpt} compares five models from four model families using identical prompts, solver, and parsers. The reference, weaker same-family, and GLM runs retain high answered accuracy. DS and QW have more shared errors. GLM produces $2$ pairs that agree on a wrong answer, and the agreement rule filters $15$ disagreements. These joint outcomes provide the retrospective error classifications used in Section~\ref{subsec:rq3}. Full-pool accuracy counts abstentions as errors.

\begin{table}[H]
\centering
\caption{Cross-model replication on L1-HARD ($240$ queries; identical
prompts, solver, and parsers on all models). Reference-grounder numbers
are shown in Table~\ref{tab:rq2-ablation}; weaker same-family, DS, and QW
rows carry Wilson $95\%$
intervals on overall accuracy. \emph{neg} $=$ accuracy on the $120$
mid-link-break negatives. DS (DeepSeek-chat) and QW
(Qwen2.5-72B-Instruct) are different vendors/architecture lineages.}
\label{tab:crossckpt}
\footnotesize
\begin{tabularx}{\columnwidth}{@{}lXrrr@{}}
\toprule
Strategy & model & overall [CI] & neg & tok/q \\
\midrule
Pure CoT   & K3   & $84.6$ & $69.2$ & $611$ \\
           & K2.6 & $82.9$ [$77.6,87.2$] & $74.0$ & $1{,}187$ \\
           & DS   & $73.3$ [$67.4,78.5$] & $46.7$ & $474$ \\
           & QW   & $90.0$ [$85.6,93.2$] & $80.8$ & $459$ \\
           & GLM  & $92.9$ [$89.0,95.5$] & $86.7$ & $1{,}627$ \\
\midrule
RAG+CoT    & K3   & $87.9$ & $75.8$ & $759$ \\
           & K2.6 & $88.3$ [$83.7,91.8$] & $80.0$ & $1{,}196$ \\
           & DS   & $84.6$ [$79.5,88.6$] & $71.7$ & $450$ \\
           & QW   & $87.1$ [$82.2,90.7$] & $75.0$ & $640$ \\
           & GLM  & $92.5$ [$88.5,95.2$] & $85.0$ & $1{,}525$ \\
\midrule
Solver 1-vote & K3 &
        $98.8$ & $99.2$ & $606$ \\
           & K2.6 & $96.3$ [$93.0,98.0$] & $99.2$ & $795$ \\
           & DS   & $83.3$ [$78.1,87.5$] & $82.5$ & $328$ \\
           & QW   & $84.6$ [$79.5,88.6$] & $95.8$ & $351$ \\
           & GLM  & $96.2$ [$93.0,98.0$] & $100$ & $1{,}118$ \\
\midrule
\textbf{CPUNeSy 2-vote} & K3 &
        $97.1$ & $98.3$ & $1{,}212$ \\
           & K2.6 & $\mathbf{94.2}$ [$90.4,96.5$] & $\mathbf{99.2}$ & $1{,}551$ \\
           & DS   & $76.7$ [$70.9,81.6$] & $77.5$ & $661$ \\
           & QW   & $82.1$ [$76.7,86.4$] & $89.2$ & $709$ \\
           & GLM  & $92.9$ [$89.0,95.5$] & $100$ & $2{,}406$ \\
\bottomrule
\multicolumn{5}{@{}p{\columnwidth}@{}}{\scriptsize served accuracy
(coverage): reference $100\%$ ($97.1\%$); weaker same-family $\mathbf{99.6\%}$
[$97.6,99.9$] ($94.6\%$), disjoint from pure $86.5\%$ [$81.4,90.3$]
and RAG $90.2\%$ [$85.7,93.4$]; GLM $\mathbf{99.1\%}$ [$96.8,99.8$]
($93.8\%$), with a confidence interval disjoint from its own pure $92.9\%$ and RAG $92.5\%$ results. GLM has $2$ agree-wrong pairs and filters $15/15$ disagreements; DS $83.6\%$ [$78.2,87.9$] ($91.7\%$)
and QW $88.0\%$ [$83.0,91.6$] ($93.3\%$), on par with their neural
baselines, consistent with the shared-error boundary of
Lemma~\ref{lem:voting-floor}. All grounded rows use the
leaf-$\Omega$ interface. DS $=$ DeepSeek-chat; QW $=$
Qwen2.5-72B-Instruct; GLM $=$ GLM-5.2.}
\end{tabularx}
\end{table}

\end{document}